\documentclass[a4paper,11pt]{scrartcl}
\usepackage[utf8x]{inputenc}
\usepackage[normalem]{ulem}
\usepackage[usenames,dvipsnames]{color}
\usepackage{microtype}
\usepackage{amsthm}
\usepackage{amssymb}
\usepackage{amsmath}
\usepackage{stmaryrd}
\usepackage{ucs,algorithmic,algorithm}
\usepackage{complexity}
\usepackage{graphicx}
\usepackage{tikz}
\usepackage{enumitem}
\usepackage{fullpage}

\makeatletter
\g@addto@macro\bfseries{\boldmath}
\makeatother

\title{Understanding model counting for $\beta$-acyclic CNF-formulas}

\newtheorem{theorem}{Theorem}

\newtheorem{definition}[theorem]{Definition}
\newtheorem{corollary}[theorem]{Corollary}
\newtheorem{lemma}[theorem]{Lemma}

\newcommand{\var}{\textsf{var}}
\newcommand{\supp}{\textsf{supp}}
\newcommand{\cla}{\textsf{cla}}

\newcommand{\N}{\mathbb{N}}
\newcommand{\Qp}{\mathbb{Q}_+}

\newcommand{\calH}{\mathcal{H}}
\newcommand{\calC}{\mathcal{C}}
\newcommand{\calT}{\mathcal{T}}

\renewcommand{\d}{\mu}

\newcommand{\tw}{\mathbf{tw}}

\newcommand{\mimw}{\mathbf{mimw}}
\newcommand{\mmw}{\mathbf{mmw}}
\newcommand{\psw}{\mathbf{psw}}

\newcommand{\sSAT}{\mathrm{\#SAT}} 
\newcommand{\mSAT}{\mathrm{MaxSAT}} 
\newcommand{\CSP}{\mathrm{CSP}} 
\newcommand{\sCSP}{\mathrm{\# CSP}} 
\newcommand{\sCSPd}{\mathrm{\# CSP_d}} 
\newcommand{\mCSPd}{\mathrm{MaxCSP_d}} 
\renewcommand{\SAT}{\mathrm{SAT}}

\newcommand{\sP}{\mathbf{\#P}}

\newcommand{\fourfrac}[5]{\frac{w(#1,#5)w(#2,#5)}{w(#3,#5)w(#4,#5)}}
\newcommand{\twofrac}[3]{\frac{w(#1,#3)}{w(#2,#3)}}
\newcommand{\twofracm}[3]{\frac{m(#1,#3)}{m(#2,#3)}}
\newcommand{\ck}[2]{#1^{(#2)}}

\PrerenderUnicode{é}

\begin{document}

\author{
Johann Brault-Baron\footnote{ LSV UMR 8643, ENS Cachan and Inria, France} \and Florent Capelli\footnote{IMJ UMR 7586  -  Logique,
Université Paris Diderot,
France} \and Stefan Mengel\footnote{ %Laboratoire d'Informatique,
LIX UMR 7161,
Ecole Polytechnique,
France}
}

\maketitle

\begin{abstract}
We extend the knowledge about so-called structural restrictions of $\sSAT$ by giving a polynomial time algorithm for $\beta$-acyclic $\sSAT$. In contrast to previous algorithms in the area, our algorithm does not proceed by dynamic programming but works along an elimination order, solving a weighted version of constraint satisfaction. Moreover, we give evidence that this deviation from more standard algorithm is not a coincidence, but that there is likely no dynamic programming algorithm of the usual style for $\beta$-acyclic $\sSAT$.
\end{abstract}

% \tableofcontents 

\section{Introduction}

The propositional model counting problem $\sSAT$ is, given a CNF-formula $F$, to count the satisfying assignments of $F$. $\sSAT$ is the canonical $\sP$-complete problem and is thus central to the area of counting complexity. Moreover, many important problems in artificial intelligence research reduce to $\sSAT$ (see e.g.~\cite{Roth96}), so there is also great interest in the problem from a practical point of view.

Unfortunately, $\sSAT$ is computationally very hard: even for very restricted CNF-formulas, e.g. monotone 2-CNF-formulas, the problem is $\sP$-hard and in fact even $\sP$-hard to approximate~\cite{Roth96}. Thus the focus of research in finding tractable classes of $\sSAT$-instances has turned to so-called \emph{structural} classes, which one gets by assigning a graph or hypergraph to a CNF-formula and then restricting the class of (hyper)graphs considered. The general idea is that if the (hyper)graph associated to an instance has a treelike decomposition that is ``nice'' enough, e.g.~a tree decomposition of small width, then there is a dynamic programming algorithm that solves $\sSAT$ for the instance. In the recent years, there has been a push for constructing such dynamic programming algorithms for ever more general classes of graphs and hypergraphs, see e.g.~\cite{FMR-08,Samer:2010wc,PaulusmaSS13,SlivovskyS13,CapelliDM14}.

Very recently, S{\ae}ther, Telle and Vatshelle, in a striking contribution~\cite{SaetherTV14}, have introduced a new width measure for CNF-formulas, that they call PS-width. Essentially, it is a measure for how much information has to be propagated from one step to the next in a natural formalization of the known dynamic programming algorithms. In our opinion, PS-width thus gives an upper bound on how far the dynamic programming techniques from the literature can be extended. Moreover, S{\ae}ther, Telle and Vatshelle have shown that if one is given a formula $F$ and a decomposition of small PS-width, one can efficiently count the number of satisfying assignments of $F$. Thus they have essentially turned the construction of dynamic programming algorithms into a question of graph theory: If, for a class of formulas, one can efficiently compute decompositions that have small PS-width for all formulas having these graphs, the dynamic programming of \cite{SaetherTV14} solves these instances. In fact, PS-width gives a uniform explanation for \emph{all} structural tractability results for $\sSAT$ from the literature that we are aware of.
On the other hand, since, in our opinion, the framework of \cite{SaetherTV14} is a very good formalization of dynamic programming, there is likely no efficient dynamic programming algorithm for a class of CNF-formulas, if it does not have decompositions of small PS-width, or if these decompositions cannot be constructed efficiently.

In this article, we focus on \emph{$\beta$-acyclic} CNF-formulas, i.e., formulas whose associated hypergraph is $\beta$-acyclic. There are several different reasonable ways of defining acyclicity of hypergraphs that have been proposed~\cite{Fagin83,Duris12}, and $\beta$-acyclicity is the most general acyclicity notion discussed in the literature for which $\sSAT$ could be tractable (see the discussions in~\cite{OrdyniakPS13,CapelliDM14}). The complexity of $\sSAT$ for $\beta$-acyclic formulas is interesting for several reasons: First, up to this paper, it was the only structural class of formulas for which we know that $\SAT$ is tractable \cite{OrdyniakPS13} without this directly generalizing to a tractability result for $\sSAT$. This is because the algorithm of \cite{OrdyniakPS13} does \emph{not} proceed by dynamic programming but uses resolution, a technique that is known to generally not generalize to counting. Moreover, $\beta$-acyclicity can be generalized to a width-measure~\cite{GottlobP01}, so there is hope that a good algorithm for $\beta$-acyclic formulas might generalize to wider classes for which even the status for $\SAT$ is left as an open problem in \cite{OrdyniakPS13}. Since decomposition techniques based on hypergraph acyclicity tend to be more general than graph-based techniques \cite{GottlobLS00}, this might lead to large, new classes of tractable $\sSAT$-instances.

The contribution of this paper is twofold: First, we show that $\sSAT$ on $\beta$-acyclic hypergraphs is tractable. In fact, we show that a more general counting problem which we call \emph{weighted counting for constraint satisfaction with default values}, short $\sCSPd$, is tractable on $\beta$-acyclic hypergraphs. We remark that there is another line of research on $\sCSP$, the counting problem related to constraint satisfaction, where dichotomy theorems for weighted $\sCSP$ depending on fixed constraint languages are proven, see e.g.~\cite{bulatov,chen}. We do \emph{not} assume that the relations of our instances are fixed but we consider them as part of the input. Thus our results and those on fixed constraint languages are completely unrelated. Instead, the structural restrictions we consider are similar to those considered e.g.~in~\cite{DalmauJ04}, but since we allow clauses, resp. relations, of unbounded arity, our results and those of \cite{DalmauJ04} are incomparable.

We note that our algorithm is in style very different from the algorithms for structural $\sSAT$ in the literature. Instead of doing dynamic programming along a decomposition, we proceed along a vertex elimination order which is more similar to the approach to $\SAT$ in \cite{OrdyniakPS13}. But in contrast to using well-understood resolution techniques, we develop from scratch a procedure to update the weights of our $\sCSPd$ instance along the elimination order. Our algorithm is non-obvious and novel, but it is relatively easy to write down and its correctness is easy to prove. Indeed most of the work in this paper is spent on showing the polynomial runtime bound which requires a thorough understanding of how the weights of instances evolve during the algorithm.

Our second contribution is that we show that our tractability result is not covered by the framework of S{\ae}ther, Telle and Vatshelle~\cite{SaetherTV14}, short STV-framework, which---as we show---covers all other known structural tractability results for $\sSAT$. We do this by showing that $\beta$-acyclic $\sSAT$-instances may have a PS-width so high that from \cite{SaetherTV14} we cannot even get subexponential runtime bounds.
This can be seen as an explanation for why the algorithm for $\beta$-acyclic $\sSAT$ is so substantially different from the algorithms from the literature. If one accepts the framework of \cite{SaetherTV14} as a good formalization of dynamic programming---which we do---then the deviation from the usual dynamic programming paradigm is not a coincidence but instead due to the fact that there \emph{is} no efficient dynamic programming algorithm in the usual style. Thus, our algorithm indeed introduces a new algorithmic technique for $\sSAT$ that allows the solution of instances that could not be solved with techniques known before. 

\section{Preliminaries and notation}

\subsection{Weighted counting for constraint satisfaction with default values}

Let $D$ and $X$ be two sets. $D^X$ denotes the set of functions from $X$ to $D$. We think of $X$ as a set of variables and of $D$ as a domain, and thus we call $a \in D^X$ an \emph{assignment} to the variables $X$. A \emph{partial assignment} to the variables $X$ is a mapping in $D^Y$ where $Y \subseteq X$. If $a \in D^X$ and $Y \subseteq X$, we denote by $a|_Y$ the \emph{restriction} of $a$ onto $Y$. For $a \in D^X$ and $b \in D^Y$, we write $a \sim b$ if $a|_{X \cap Y} = b|_{X \cap Y}$ and if $a \sim b$, we denote by $a \cup b$ the mapping in $D^{X \cup Y}$ with $(a \cup b)(x) = a(x)$ if $x \in X$ and $(a \cup b)(x) = b(x)$ otherwise. Let $a \in D^X$, $y \notin X$ and $d \in D$. We write $a \oplus_y d := a \cup \{y \mapsto d\}$.

% We are now ready to define weighted $\CSP$s. It is a natural generalization of $\CSP$s where each constraint have a value depending on the assignment of the variables. 

\begin{definition}
A {\em weighted constraint with default value} $c = (f,\d)$ on variables $X$ and domain $D$ is a function  $f: S \rightarrow \Qp$  with $S \subseteq D^X$ and $\d \in \Qp$. $S = \supp(c)$ is called the {\em support} of $c$, $\d(c) = \d$ its default value and we denote by $\var(c) = X$ the variables of $c$. We define the size $|c|$ of the constraint $c$ to be $|c|:= |S|\cdot |\var(c)|$. The constraint $c$ naturally induces a total function on $D^X$, also denoted by $c$, defined by $c(a) = f(a)$ if $a \in S$ and $c(a) = \d$ otherwise.
\end{definition}

Observe that we do not assume $\var(c)$ to be non-empty. A constraint whose set of variables is empty has only one possible value in its support: the value associated to the empty assignment (the assignment that assigns no variable).

Since we only consider weighted constraints with default value in this paper, we will only say \emph{weighted constraint} where the default value is always implicitly understood. Note that we have to restrict ourselves to non-negative weights, because non-negativity will be crucial in the proofs. This is not a problem in our context, non-negative numbers are sufficient to encode $\sSAT$ as we will see in Section~\ref{sct:otherproblems}.

\begin{definition}
The problem $\sCSPd$ is the problem of computing, given a finite set $I$ of weighted constraints on domain $D$, the {\em partition function} $$w(I) = \sum_{a \in D^{\var(I)}} \prod_{c \in I} c(a|_{\var(c)}),$$
where $\var(I) := \bigcup_{c \in I} \var(c)$. 

The {\em size} $\|I\|$ of a $\sCSPd$-instance $I$ is defined to be $\|I\| := \sum_{c \in I} |c|$. Its {\em structural size} $s(I)$ of $I$ is defined to be $s(I) := \sum_{c \in I} |\var(c)|$.
\end{definition}

Note that the size of an instance as defined above roughly corresponds to that of an encoding in which the non-default values, i.e., the values on the support, are given by listing the support and the associated values in one table for each relation. We consider this convention as very natural and indeed it is near to the conventions in database theory and artificial intelligence.

Given an instance $I$, it will be useful to refer to subinstances of $I$, that is a set $J \subseteq I$. We will also refer to partition function of subinstances under some partial assignment, that is, the partition function of $J$ where some of its variables are forced to a certain value. To this end, for $a \in D^W$, with $W \subseteq \var(I)$, and $J \subseteq I$ with $V' = \var(J)$ we define
$$ w(J,a) := \sum_{b \in D^{V'} \atop a \sim b} \prod_{c \in J} c(b|_{\var(c)}).$$

We use the following straightforward observation throughout this paper
$$ w(J,a) = \sum_{b \in D^{V' \setminus W}} \prod_{c \in J} c((a \cup b)|_{\var(c)}).$$

\subsection{Graphs and hypergraphs associated to CNF-formulas}

We use standard notation for graphs which can e.g.~be found in \cite{Diestel05}.
A {\em hypergraph} $\calH=(V,E)$ consists of a finite set $V$ and a set $E$ of non-empty subsets of $V$. The elements of $V$ are called \emph{vertices} while the elements of $E$ are called \emph{edges}. As usual for graphs, we sometimes denote the vertex set of a hypergraph $\calH$ by $V(\calH)$ and the edge set of $\calH$ by $E(\calH)$.
The size of a hypergraph is defined to be $\|\calH\| = \sum_{e \in E(\calH)} |e|$. 

A {\em subhypergraph} $\calH'$ of a hypergraph $\calH$ is a hypergraph such that $V(\calH')\subseteq V(\calH)$ and $E(\calH')\subseteq \{e\cap V(\calH') \mid e\in E(\calH), e\cap V(\calH') \ne \emptyset\}$. For $S \subseteq V(\calH)$, the {\em induced subhypergraph} of $\calH$ by $S$ is the hypergraph $\calH[S] = (S, \{ e \cap S \mid e \in E(\calH) \} \setminus \{ \emptyset \})$. We denote by $\calH\setminus S$  the hypergraph $\calH[V(\calH)\setminus S]$. If $S$ contains only one vertex $v$, we also write $\calH\setminus v$ for $\calH\setminus \{v\}$.

We are interested in structural restrictions of the problem $\sCSPd$. What we mean by structural restriction is that we restrict the way the variables interact in the different constraints. To formalize this notion, we introduce the hypergraph associated to an instance of $\sCSPd$:
The hypergraph $\calH(I)$ associated to $\sCSPd$-instance $I$ is the hypergraph $\calH(I) := (\var(I), E_I)$ where $E_I := \{\var(c) \mid c \in I\}$. The hypergraph of a CNF-formula is defined as $\calH(F):= (\var(F), E_F)$ where $E_F := \{\var(C) \mid C \in \cla(F)\}$ where $\var(F)$ denotes the set of variables of $F$ and $\cla(F)$ denotes the set of clauses of $F$.

The incidence graph $I(\calH)$ of a hypergraph $\calH=(V,E)$ is the bipartite graph with the vertex set $V\cup E$ and an edge between $v\in V, e\in E$ if and only if $v\in e$. Similarly, the incidence graph $I(F)$ of a CNF-formula $F$ has the vertex set $\var(F)\cup \cla(F)$ and $x\in \var(F)$ and $C\in \cla(F)$ are connected  by an edge if and only if $x$ appears in $C$. 

\subsection{Relation to $\sSAT$}\label{sct:otherproblems}

We show in this section how we can encode $\sSAT$ into $\sCSPd$-instances with the same hypergraphs.

The problem $\SAT$ differs from the classical $\CSP$ framework in the way the constraints are represented. Classically, in $\CSP$, all the solutions to a constraint are explicitly listed. For a CNF-formula however, each clause with $n$ variables has $2^n-1$ solutions, which would lead to a $\CSP$-representation exponentially bigger than the CNF-formula. One way of dealing with this is encoding CNF-formulas into $\CSP$-instances by listing all assignments that are {\em not} solution of a constraint, see e.g.~\cite{CohenGH09}. In this encoding, each clause has only one counter-example and the corresponding $\CSP$-instance is roughly of the same size as the CNF-formula.% \stefan{shorten this at least for the conference version?}

The strength of the $\CSP$ with default values is that it can easily embed both representations. This leads to a polynomial reduction of $\sSAT$ to $\sCSPd$.
\begin{lemma}
  \label{lem:sattocsp}
  Given a CNF-formula $F$ one can construct in polynomial time a $\sCSPd$-instance $I$ on variables $\var(F)$ and domain $\{0,1\}$ such that
  \begin{itemize}
  \item $\calH(F) = \calH(I)$,
  \item for all $a \in \{0,1\}^{\var(F)}$, $a$ is a solution of $F$ if and only if $w(I,a) = 1$, and otherwise $w(I,a) = 0$, and
  \item $s(I) = \|I\| = |F|$.
  \end{itemize}
\end{lemma}
\begin{proof}
  For each clause $C$ of $F$, we define a constraint $c$ with default value $1$ whose variables are the variables of $C$ and such that $\supp(c) = \{a\}$ and $c(a) = 0$, where $a$ is the only assignment of $\var(C)$ that is not a solution of $C$. It is easy to check that this construction has the above properties.
\end{proof}

\subsection{$\beta$-acyclicity of hypergraphs}

In this section we introduce the characterizations of $\beta$-acyclicity of hypergraphs we will use in this paper. We remark that there are many more characterizations, see e.g.~\cite{BrandstaedtLS99,JBB14}.

\begin{definition}
  Let $\calH$ be a hypergraph. A vertex $x \in V(\calH)$ is defined to be a {\em nest point} if $ \{e \in E(\calH) \mid x\in e\}$ forms a sequence of sets increasing for inclusion, that is $\{e \in E(\calH) \mid x\in e\} = \{e_1, \ldots, e_k\}$ with $e_i \subseteq e_{i+1}$ for $i\in \{1, \ldots, k-1\}$.

A {\em $\beta$-elimination order} for $\calH$ is defined inductively as follows:
\begin{itemize}
\item If $\calH = \emptyset$, then only the empty tuple is a $\beta$-elimination order for $\calH$.
\item Otherwise, $(x_1, \ldots, x_n)$ is a $\beta$-elimination for $\calH$ if $x_1$ is a nest point of $\calH$ and $(x_2,\ldots,x_{n})$ is a $\beta$-elimination order for $\calH\setminus x_1$.
A hypergraph $\calH$ called $\beta$-acyclic if and only if there exists a $\beta$-elimination order for $\calH$.
\end{itemize}
\end{definition}

It is easy to see that one can test $\beta$-acyclicity of a graph in polynomial time and that one can compute a $\beta$-elimination order efficiently if it exists. 

We will also make use of another equivalent characterization of $\beta$-acyclic hypergraphs. A graph $G$ is defined to be \emph{chordal bipartite} if it is bipartite and every cycle of length at least $6$ in $G$ has a chord.

\begin{theorem}[\cite{AusielloDM86}]\label{thm:chorbipbeta}
 A hypergraph is $\beta$-acyclic if and only if its incidence graph is chordal bipartite.
\end{theorem}

We say that a $\sCSPd$-instance $I$ is $\beta$-acyclic if $\calH(I)$ is $\beta$-acyclic and we use an analogous convention for $\sSAT$. Note that the incidence graph of an instance $I$ and that of its hypergraph in general do not coincide, because $I$ might contain several constraints with the same sets of variables. But with Theorem~\ref{thm:chorbipbeta}, it is not hard to see that the incidence graph of an instance $I$ is chordal bipartite if and only if the incidence graph of the hypergraph of $I$ is chordal bipartite, so we can interchangeably use both notions of incidence graphs in this paper without changing the class of instances.

\begin{corollary}
  \label{cor:sattocsp}
  $\sSAT$ is polynomial time reducible to $\sCSPd$. Moreover, $\sSAT$ restricted to $\beta$-acyclic formulas is polynomial time reducible to $\sCSPd$ restricted to $\beta$-acyclic instances. 
\end{corollary}
\begin{proof}
Taking the construction of Lemma~\ref{lem:sattocsp}, it is clear that the number of solution of $F$ is equal to $w(I)$. The rest follows from the fact that the hypergraph remains unchanged during the reduction.
\end{proof}

\subsection{Width measures of graphs and CNF-Formulas}

In this section we introduce several width measures on graphs and CNF-formulas that are used when relating our algorithm for $\beta$-acyclic $\sCSPd$ to the framework of S{\ae}ther, Telle and Vatshelle~\cite{SaetherTV14}. Readers only interested in the algorithmic part of this paper may safely skip to Section~\ref{sec:cspnest}.

We consider several width notions that are mainly defined by branch decompositions. For an introduction into this area and many more details see~\cite{Vatshelle12}.
For a tree $T$ we denote by $L(T)$ the set of the leaves of $T$. A \emph{branch decomposition} $(T,\delta)$ of a graph $G=(V,E)$ consists of a subcubic tree $T$, i.e., a tree in which every vertex has at most degree $3$, and a bijection $\delta$ between $L(T)$ and $V$. For convenience we often identify $L(T)$ and $V$. Moreover, it is often convenient to see a branch decomposition as rooted tree, and as this does not change any of the notions we define (see~\cite{Vatshelle12}), we generally follow this convention. For every $x\in V(T)$ we define $T_x$ be the subtree of $T$ rooted in $x$. From $x$ we get a partition or \emph{cut} of $V$ into two sets defined by $(L(T_x), V \setminus L(T_x))$. For a set $X\subseteq V$ we often write $\bar{X}$ for $V\setminus X$.

Given a symmetric function $f : 2^V\times 2^V \rightarrow \mathbb{R}$ we define the $f$-width of a branch decomposition $(T, \delta)$ to be $\max_{x\in V(T)} f(L(T_x), V \setminus L(T_x))$, i.e., the $f$-width is the maximum value of $f$ over all cuts of the vertices of $T$. The $f$-branch width of a graph $G$ is defined as the minimum $f$-width of all branch decompositions of $G$.

Given a graph $G=(V,E)$ and a cut $(X, \bar{X})$ of $V$, we define $G[X, \bar{X}]$ to be the graph with vertex set $V$ and edge set $\{uv\mid u\in X, v\in \bar X, uv\in E\}$. 

%\stefan{move this in conference version?}
We will use at several places the well-known notion of treewidth of a graph $G$, denoted by $\tw(G)$. Instead of working with the usual definition of treewidth (see e.g.~\cite{Bodlaender93}), it is more convenient for us to work with the strongly related notion of \emph{Maximum-Matching-width} (short \emph{MM-width}) introduced by Vatshelle~\cite{Vatshelle12}. The MM-width of a graph $G$, denoted by $\mmw(G)$, is defined as the $f$-branch width of $G$ for the function $f$ that, given a cut $(X, \bar{X})$ of $G$, computes the size of the maximum matching of $G[X, \bar X]$.
MM-width and treewidth are linearly related \cite[p. 28]{Vatshelle12}.

\begin{lemma}\label{lem:MMvsTW}
 Let $G$ be a graph, then $\frac{1}{3} (\tw(G)+1) \le \mmw(G) \le \tw(G)+1$.
\end{lemma}

Another width measure of graphs that we will use extensively is \emph{Maximum-Induced-Matching-width} (short \emph{MIM-width}): The MIM-width of a graph $G$, denoted by $\mimw(G)$, is defined as the $f$-branch width of $G$ for the function $f$ that, given a cut $(X, \bar{X})$ of $G$, computes the size of the maximum induced matching of $G[X, \bar X]$. 

Given a CNF-formula~$F$, we say that a set of clauses $\calC \subseteq \cla(F)$ is \emph{precisely satisfiable} if there is an assignment to $F$ that satisfies all clauses in $\calC$ and no clause in $\cla(F)\setminus \calC$. The PS-value of $F$ is defined to be the number of precisely satisfiable subsets of $\cla(F)$. Let $F$ be a CNF-formula, $X\subseteq \var(F)$ and $\calC \subseteq \cla(F)$. Then we denote by $F_{X,\calC}$ the formula we get from $F$ by deleting first every clause not in $\calC$ and then every variable not in $X$.

Let $I(F)$ be the incidence graph of $F$ and let $(A, \bar{A})$ be a cut of $I(F)$. Let $X:= \var(F) \cap A$, $\bar X := \var(F) \cap \bar A$, $\calC := \cla(F) \cap A$ and $\bar \calC := \cla(F)\cap \bar A$. Let $ps(A,\bar A)$ be the maximum of the PS-values of $F_{X, \bar \calC}$ and $F_{\bar X, \calC}$. Then the PS-width of a branch decomposition $(T, \delta)$ of $G$ is defined as the $ps$-branch width of $(T, \delta)$. Moreover, the PS-width of $F$, denoted $\psw(F)$, is defined to be the $ps$-branch width of $I(F)$.

Let us try to give an intuition why we believe that PS-width is a good notion to model the limits of tractable dynamic programming for $\sSAT$: The dynamic programming algorithms in the literature typically proceed by cutting instances into subinstances and then iteratively solving the instance along these cuts. During this process, some information has to be propagated between the subinstances. Intuitively, a minimum amount of such information is which sets of clauses are already satisfied by certain assignments and which clauses still have to be satisfied later in the process. In doing this, the individual clauses can be ``bundled together'' if they are satisfied by an assignment simultaneously. The number such bundles is exactly the PS-width of a cut, so we feel that PS-width is a good formalization of the minimum amount of information that has to be propagated during dynamic programming in the style of the algorithms from the literature. 

Not only is PS-width in our opinion a good measure for the limits of dynamic programming, but S{\ae}ther, Telle and Vatshelle also showed that it allows tractable solving of $\sSAT$.

\begin{theorem}[\cite{SaetherTV14}]\label{thm:dynamic}
 Given a CNF-formula $F$ of $n$ variables and $m$ clauses and of size $s$, and a branch decomposition $(T,\delta)$ of the incidence graph $I(F)$ of $F$ with PS-width $k$, one can count the number of satisfying assignments of $F$ in time $O(k^3 s(m+n))$.
\end{theorem}

We admit that the intuition explained above is rather vague and informal, so the reader might or might not share it, but we stress that it is supported more rigorously by the fact that all known tractability results from the literature that were shown by dynamic programming can be explained by a combination of PS-width and Theorem~\ref{thm:dynamic}.

S{\ae}ther, Telle and Vatshelle showed the following connection between the PS-width of a CNF-formula~$F$ and the MIM-width of the incidence graph $G$ of $F$.

\begin{theorem}[\cite{SaetherTV14}]\label{thm:MIMvsPS} 
For any CNF-formula $F$ over $m$ clauses, any branch decomposition of the incidence graph $I(F)$ of $F$ with MIM-width $k$ has PS-width at most $m^{k}$.
\end{theorem}

Theorem \ref{thm:MIMvsPS} and Theorem \ref{thm:dynamic} essentially turn finding structural classes of tractable $\sSAT$-instances into a problem of graph theory: it suffices to show that certain classes of formulas have sufficiently small MIM-width or PS-width to show that they are tractable. We will see that all tractability results from the literature can be explained this way. Unfortunately, deciding if a class of formulas has small MIM-width or PS-width seems to be tricky. In fact, even the complexity of deciding if a given graph has MIM-width $1$ in polynomial time is left as an open problem in \cite{Vatshelle12}.

\section{The algorithm and its correctness}\label{sec:cspnest}

In this section we describe an algorithm that, given an instance $I$ of $\sCSPd$ on domain $D$ and a nest point $x$ of $\calH(I)$, constructs in a polynomial number of arithmetic operations an instance $I'$ such that $\calH(I') = \calH(I)\setminus x$, $\|I'\| \leq \|I\|$ and $w(I) = |D|w(I')$. We then explain that if $I$ is $\beta$-acyclic, we can iterate the procedure to compute $w(I)$ in a polynomial number of arithmetic operations.

In the following, for $x \in \var(I)$, we denote by $I(x) = \{c \in I \mid x \in \var(c) \}$.

\begin{theorem}
\label{thm:cspalgo}
  Let $I$ be a set of weighted constraints on domain $D$ and $x$ a nest point of $\calH(I)$. Let $I(x) = \{c_1, \ldots, c_p\}$ with $\var(c_1) \subseteq \ldots \subseteq \var(c_p)$. Let $I' = \{c' \mid c \in I\}$ where
  \begin{itemize}
  \item if $c \notin I(x)$ then $c' := c$
  \item if $c = c_i$, then $c_i' := (f_i',\d)$ is the weighted constraint on variables $\var(c_i') = \var(c) \setminus \{x\}$, with default value $\d(c_i)$ and $\supp(c_i') := \{a \in D^{\var(c_i')} \mid \exists d \in D, (a \oplus_x d) \in \supp(c_i)\}$. Moreover, for all $a \in \supp(c_i')$, let 
  $P_i(a,d) := \prod_{j=1}^i c_j((a \oplus_x d)|_{\var(c_j)})$ and $P_0(a,d) = 1$. We define:
$$ f_i'(a) := \frac{\sum_{d \in D} P_{i}(a,d)}{\sum_{d \in D} P_{i-1}(a,d)}$$
if $\sum_{d \in D} P_{i-1}(a,d) \neq 0$ and $f_i'(a) := 0$ otherwise.
\end{itemize}
Then $\calH(I') = \calH(I)\setminus x$, $\|I'\| \leq \|I\|$ and $w(I) = |D|w(I')$. Moreover, one can compute $I'$ with a $O(p\|I(x)\|)$ arithmetic operations.
\end{theorem}
\begin{proof}
  First, we explain why $I'$ is well-defined. As $x$ is a nest point, we can write $I(x) = \{c_1, \ldots, c_m\}$ with $\var(c_1) \subseteq \ldots \subseteq \var(c_m)$. If two constraints have the same variables, we choose an arbitrary order for them. Note that in Section~\ref{sec:analysis} we will choose a specific order that ensures that the algorithm runs in polynomial time on a RAM, but in this proof any order will do. Finally, remark that $P_i(a,d)$ is well defined since for $a \in \supp(c_i')$, $d \in D$ and $j \leq i$, $(a \oplus_x d)$ assigns all variables of $c_j$ since $\var(c_j) \subseteq \var(c_i)$. Thus writing $c_j((a \oplus_x d)|_{\var(c_j)})$ is correct. We insist on the fact that it is only because $x$ is a nest point that this definition works.

  $\calH(I') = \calH(I)\setminus x$ is obvious because for a constraint in $I$ with variable set $V$, there exists a constraint in $I'$ with variable set $V \setminus \{x\}$.

  $\|I'\| \leq \|I\|$ because for all $c \in I$, $|c'| \leq |c|$ since $|\{a \in D^{\var(c')} \mid \exists d \in D, (a \oplus_x d) \in \supp(c)\}| \leq |\supp(c)|$.

We now show by induction on $i$ that for all $a \in D^{\var(c_i')}$, 
$$ |D| \prod_{j = 1}^i c_j'(a) = \sum_{d \in D} P_i(a,d).$$

For $i=1$, let $a \in D^{\var(c_1')}$. If $a \in \supp(c_1')$, then by definition:
$$ c_1'(a) = \frac{\sum_{d \in D} P_1(a,d)}{\sum_{d \in D} P_0(a,d)}.$$
Since $P_0(a,d) = 1$ for all $d$, we have the expected result. 

If $a \notin \supp(c_1')$, then for all $d$, $a \oplus_x d \notin \supp(c_1)$. Thus $P_1(a,d) = \d_1$ for all $d$ and finally 
$$ \sum_{d \in D} P_1(a,d) = |D|\d_1 = |D|c_1'(a).$$

Now suppose that the result holds for $i$. Let $a \in D^{\var(c_i')}$. Then we get by induction
$$ |D| \prod_{j=1}^{i+1} c_j'(a) = (\sum_{d \in D} P_i(a,d))c_{i+1}'(a).$$

First, assume that $\sum_{d \in D} P_i(a,d) = 0$. Since this sum is a sum of positive rationals, we have that for all $d$, $P_i(a,d) = 0$. Thus, $P_{i+1}(a,d) = 0$ for all $d$, that is $\sum_{d \in D} P_{i+1}(a,d) = 0$ which confirm the induction hypothesis.

Now assume that $\sum_{ d\in D} P_i(a,d) \neq 0$. If $a \in \supp(c_{i+1}')$, by definition of $c_{i+1}'$, the induction hypothesis trivially holds.

If $a \notin \supp(c_{i+1})$, we have $P_{i+1}(a,d) = \d_{i+1}P_i(a,d)$ for all $d$. Thus $\sum_{d \in D} P_{i+1}(a,d) = \d_i \sum_{d \in D} P_i(a,d) = c'_{i+1}(a) \sum_{d \in D} P_i(a,d)$ which establish the induction hypothesis for $i+1$. 

Applying the result for $i = p$, we find:

$$ |D| \prod_{c \in I(x)} c'(a) = \sum_{d \in D} \prod_{c \in I(x)} c((a \oplus_x d)|_{\var(c)})$$

Now, it is sufficient to remark that for $c \notin I(x)$, for all $d \in D$, $c((a \oplus_x d)|_{\var(c)}) = c(a|_{\var(c)}) = c'(a|_{\var(c)})$ since $x \notin \var(c)$ and $c = c'$. Thus:

$$ |D|w(I') = \sum_{a \in D^{\var(I) \setminus \{x\}}} \sum_{d \in D} \prod_{c \in I'} c((a \oplus_x d)|_{\var(c)}) = w(I).$$

We now analyze the number of arithmetic operations we make in the construction of $I'$. Clearly, if we have computed the $\sum_{d \in D} P_i(a,d)$ for all $i \leq p$ and $a \in \supp(c_i')$ then we can compute $c_i'(a)$ with one division. Thus we need to do $p$ divisions. Now remark that if we have computed $P_i(a,d)$, then we only need one more multiplication to compute $P_{i+1}(a,d)$. 

Now, we prove by induction on $i$ that $P_i(a,d)$ could take at least $1 + \sum_{j=1}^i |c_j|$ different values. It is trivial for $i = 0$. Now remark that if $a \oplus_x d \notin \supp(c_i)$, then $P_i(a,d) = \d_i P_{i-1}(a,d)$, thus by induction, it gives $1 + \sum_{j=1}^{i-1}|c_j|$ different values for $P_i$. And there is at most $|\supp(c_i)| \leq |c_i|$ other values for $a \oplus_x d \in \supp(c_i)$, which prove the induction.

In the end, we have to compute at most $O(p \times \|I(x)\|)$ different values for the $P_i$ which can be done with a $O(p \times \|I(x)\|)$ multiplications. Now if $i$ is fixed, for all $a$, $\sum_{d \in D} P_i(a,d)$ have at most $1 + \sum_{j=1}^i |c_j|$ different terms that are already computed. Thus we only need $O(\|I(x)\|)$ operations to compute each of them. As there is $p$ different sums to compute, we can do everything with a $O(p\|I(x)\|)$ arithmetic operations.
\end{proof}

\begin{theorem}
\label{thm:csparith}
  If $I$ is a $\beta$-acyclic instance of $\sCSPd$, we can compute $w(I)$ with a $O(s(I)^2\|I\|)$ arithmetic operations.
\end{theorem}
\begin{proof}
We iterate the algorithm of Theorem~\ref{thm:cspalgo} on a $\beta$-elimination order of the variables of $I$ to transform it into an instance $I^*$. After all variables are eliminated, every constraint of $I^*$ has an empty set of variables, thus $w(I^*) = \prod_{c \in I^*} c(\epsilon)$, where $\epsilon$ denotes the empty assignment. Moreover, by Theorem~\ref{thm:cspalgo}, $w(I) = |D|^{|\var(I)|}w(I^*)$. Thus $w(I)$ can be computed with $O(s(I))$ additionnal multiplications.

If we denote by $p_x = |\{c \in I \mid x \in \var(c)\}|$, we have a total complexity of $\sum_{x \in \var(I)} O(p_x\|I(x)\|)$, that is $O((\sum_{x \in \var(I)} p_x) |\var(I)| \|I\|)$. It is easy to see that $\sum_{x \in \var(I)} p_x = s(I)$ and since $|\var(I)| \leq s(I)$, we have a total number of arithmetic operations that is a $O(s(I)^2 \|I\|)$.
\end{proof}

\section{Runtime analysis of the algorithm}\label{sec:analysis}

The analysis of Theorem~\ref{thm:csparith} shows that our algorithm uses only a polynomial number of arithmetic operations. Unfortunately, this does not guarantee that the algorithm runs in polynomial time on a RAM. The problem is that, due to the many multiplications and divisions, the bitsize of the new (rational) weights computed by the algorithm at each step could grow exponentially, leading to an overall superpolynomial runtime. In this section we will prove that this is in fact not the case. We will show that at each step of the algorithm, numerous cancellations occur, leading to weights of polynomial bitsize. Combining this with Theorem~\ref{thm:csparith}, it will follow that the algorithm runs in polynomial time.

\subsection{Some technical lemmas}
\label{sec:techlem}
In this section, we will show some rather technical lemmas we will use later on. Throughout this paper, we follow the convention that for all assignment $a$, we have $w(\emptyset,a) = 1$. This is motivated by the following lemma.

\begin{lemma}
  \label{lem:sub0}
  Let $I$ be a set of weighted constraints, $J \subseteq I$ and $a$ a partial assignment of $\var(I)$. If $w(J,a) = 0$ then $w(I,a) = 0$.
\end{lemma}
\begin{proof}
  We have $w(J,a) = 0 = \sum_{b \simeq a} \prod_{c \in J} c(b|_{\var(c)})$. Since every term of the sum is non-negative, we have that for all $b \simeq a$ it holds $\prod_{c \in J} c(b|_{\var(c)}) = 0$. Thus, $$w(I,a) = \sum_{b \simeq a} \prod_{c \in J} c(b|_{\var(c)}) \prod_{c \in I \setminus J} c(b|_{\var(c)}) = 0.$$
\end{proof}

One key ingredient in our analysis will be understanding how two subinstances interact under a partial assignment. 

\begin{lemma}
\label{lem:cspmult}
Let $I$ be a set of weighted constraints on domain $D$, $J_1 \subseteq I$, $J_2 \subseteq I$ and $a \in D^W$ for $W\subseteq \var(I)$. Let $V_1 = \var(J_1)$, $V_2 = \var(J_2)$. If $V_1 \cap V_2 \subseteq W$ and $J_1 \cap J_2 = \emptyset$, then 
$$ w(J_1 \cup J_2, a) = w(J_1,a)w(J_2,a)$$
\end{lemma}
\begin{proof}
Let $V = V_1 \cup V_2$. Since $V \setminus W = (V_1 \setminus W) \cup (V_2 \setminus W)$ and this union is disjoint by definition, there is a natural bijection between $D^{V_1 \setminus W} \times D^{V_2 \setminus W}$ and $D^{V \setminus W}$ that associates to $(b_1,b_2)$ the assignment $b_1 \cup b_2$. Moreover, if $c \in J_1$, then $(b_1 \cup b_2)|_{\var(c)} = b_1|_{\var(c)}$ since $\var(c) \subseteq V_1$. Similarly, for $c \in J_2$, $(b_1 \cup b_2)|_{\var(c)} = b_2|_{\var(c)}$. Consequently,
$$ 
\begin{aligned}
w(J_1 \cup J_2, a) & = \sum_{b_1 \in D^{V_1 \setminus W}} \sum_{b_2 \in D^{V_2 \setminus W}}\prod_{c \in J_1} c((a \cup b_1)|_{\var(c)}) \prod_{c \in J_2} c((a \cup b_2)|_{\var(c)}) \\
& = w(J_1,a)w(J_2,a)
\end{aligned}
$$
\end{proof}

\begin{corollary}
\label{cor:cspdiv}
Let $I$ be a set of weighted constraints on domain $D$, $J_1 \subseteq I$, $J_2 \subseteq I$ and $a \in D^W$ for $W\subseteq \var(I)$. Let $V_1 = \var(J_1 \setminus J_2)$, $V_2 = \var(J_2 \setminus J_1)$ and $V_0 = \var(J_1 \cap J_2)$. If $V_0 \cap V_1 \subseteq W$ and $V_0 \cap V_2 \subseteq W$. If $w(J_2,a) \neq 0$, we have:
$$ \twofrac{J_1}{J_2}{a} = \twofrac{J_1 \setminus J_2}{J_2 \setminus J_1}{a}$$
\end{corollary}
\begin{proof}
First, remark that $w(J_2 \setminus J_1, a) \neq 0$ by Lemma~\ref{lem:sub0} since $J_2 \setminus J_1 \subseteq J_2$ and $w(J_2,a) \neq 0$.

Apply Lemma~\ref{lem:cspmult} on $J_1 \setminus J_2$ and $J_1 \cap J_2$ for the numerator and on $J_2 \setminus J_1$ and $J_2 \cap J_1$ for the denominator and observe that $w(J_1 \cap J_2,a)$ cancels.
\end{proof}

We will use the following corollary heavily in Section~\ref{sec:analysis}.
\begin{corollary}
\label{cor:cspcancel}
Let $I$ be a set of weighted constraints on domain $D$, $J_1,J_2,J_3,J_4 \subseteq I$ and $a \in D^W$ for $W\subseteq \var(I)$. Assume that $w(J_3,a) \neq 0$ and $w(J_4,a) \neq 0$ and
\begin{enumerate}[label=(\roman{enumi})]
\item \label{condi} $J_1 \cap J_2 \subseteq J_3$ and $J_3 \cap J_4 \subseteq J_1$,
\item \label{condii} $\var(J_1 \setminus J_3) \cap \var(J_1 \cap J_3) \subseteq W$,
\item \label{condiii} $\var(J_3 \setminus J_1) \cap \var(J_1 \cap J_3) \subseteq W$,
\item \label{condiv} $\var(J_1 \setminus J_3) \cap \var(J_2) \subseteq W$, and
\item \label{condv} $\var(J_3 \setminus J_1) \cap \var(J_4) \subseteq W$.
\end{enumerate}
Then $$ \fourfrac{J_1}{J_2}{J_3}{J_4}{a} = \twofrac{(J_1 \setminus J_3) \cup J_2}{(J_3 \setminus J_1) \cup J_4}{a}.$$
\end{corollary}
\begin{proof}
Apply Corollary~\ref{cor:cspdiv} on $J_1$ and $J_3$ and Lemma~\ref{lem:cspmult} on $J_1 \setminus J_3$ and $J_2$ for the numerator and on $J_3 \setminus J_1$ and $J_4$ for the denominator. Remark that Condition~\ref{condi} ensures that $(J_1 \setminus J_3) \cap J_2 = \emptyset$ and $(J_3 \setminus J_1) \cap J_4 = \emptyset$ and that the denominator is not null because $w((J_3 \setminus J_1) \cup J_4,a) = w(J_3 \setminus J_1,a)w(J_4,a)$ and $w(J_4,a) \neq 0$ by assumption and $w(J_3 \setminus J_1,a) \neq 0$ by Lemma~\ref{lem:sub0} and $w(J_3,a) \neq 0$.
\end{proof}

\subsection{Defining partial orders}
\label{sec:partorder}

The algorithm of Theorem~\ref{thm:cspalgo} transforms an instance into a new one with the same number of constraints but with one variable less. In this section we will give an explicit description of the weight of a constraint $c \in I$ after $k$ such elimination steps. In the following, $I$ is a $\beta$-acyclic $\CSP$-instance and $\{x_1, \ldots, x_n\} = \var(I)$ is a $\beta$-elimination order of $\calH(I)$. We assume that we will perform the elimination along this order. Let $X_k = \{x_1, \ldots, x_k\}$ and for $c \in I$, we denote by $\ck{c}{k}$ the constraint $c$ after the elimination of $x_k$. By convention, $\ck{c}{0} = c$. Remark that $\var(\ck{c}{k}) = \var(c) \setminus X_k$. 

In the following, we will introduce for each $k$, a partial order $\prec_k$ on $I$. The intuition for this partial order is that for $c,d \in I$, $c \prec_k d$ means that $\ck{d}{k}$ ``depends on'' $\ck{c}{0}$. For example, assume that $x_1 \in \var(c) \subseteq \var(d)$. When we eliminate $x_1$, we see---in the formula of Theorem~\ref{thm:cspalgo}---that the weight of $c$ appears in the definition of $\ck{d}{1}$. Hence, we would like to have $c \prec_1 d$. 

To simplify the proofs, we make one more assumption on $I$: If $c,d \in I$ and $c \neq d$, then $\var(c) \neq \var(d)$. We may assume this w.l.o.g.~since it is easy to merge two constraints with the same variables without increasing $\|I\|$. Observe that we make this assumption only on the initial instance $I$. During the elimination process, constraints with the same set of variables might appear, but we can easily handle them. 

\begin{definition}
  For two constraints $c,d \in I$, we write $c \prec d$ if there exists $k$ such that $\var(c) \setminus X_k \subsetneq \var(d) \setminus X_k$. We write $c \preceq d$ if $c\prec d$ or $c= d$.
\end{definition}

\begin{lemma}\label{lem:totalorder}
  $\preceq$ is a total order on $I$.
\end{lemma}
\begin{proof}
We first show that $\preceq$ is antisymmetric. So let $c,d$ be constraints such that $c \preceq d$ and $d\preceq c$. By way of contradiction, assume that $c\ne d$, so $c\prec d$ and $d\prec c$. By definition there are $k, k'$ such that $\var(c) \setminus X_k \subsetneq \var(d) \setminus X_k$ and $\var(d) \setminus X_{k'} \subsetneq \var(c) \setminus X_{k'}$. W.l.o.g.~assume that $k<k'$. Then $\var(c) \setminus X_\ell \subseteq \var(d) \setminus X_\ell$ for all $\ell \ge k$ which is a contradiction to $d\prec c$. It follows that $\preceq$ is antisymmetric.
  
We now show transitivity of $\preceq$. So let $c,d,e \in I$ with $c \preceq d$ and $d \preceq e$. If we have $c=d$ or $d=e$, then we get immediately $c\preceq d$. Thus we may assume that $c\prec d$ and $d\prec e$. By definition, there exist $k,\ell$ such that $\var(c) \setminus X_k \subsetneq \var(d) \setminus X_k$ and $\var(d) \setminus X_\ell \subsetneq \var(e) \setminus X_\ell$. For $m := \max(k,\ell)$ we get $\var(c) \setminus X_m \subseteq \var(d) \setminus X_m \subseteq \var(e) \setminus X_m$ and one of these inclusions is strict. Thus $\var(c) \setminus X_m \subsetneq  \var(e) \setminus X_m$, that is $c \prec e$ and it follows that $\preceq$ is transitive.

We now show that $\preceq$ is total. So let $c,d \in I$. If $c=d$, then by definition $c\preceq d$. So we assume that $c\ne d$. Let $k = \max \{j \mid x_j (\in \var(c) \setminus \var(d)) \cup (\var(d) \setminus \var(c))\}$. Observe that $k$ is well-defined, since $\var(d) \neq \var(c)$ by assumption on $I$. Assume first that $x_k \in \var(d) \setminus \var(c)$. Then $\var(c) \setminus \var(d) \subseteq X_{k-1}$ by maximality of $k$. It follows that  $\var(c) \setminus X_{k-1} \subseteq \var(d) \setminus X_{k-1}$ and since $x_k\in \var(d) \setminus X_{k-1}$, we have $\var(c) \setminus X_{k-1} \subsetneq \var(d) \setminus X_{k-1}$. Thus $c \prec d$. Analogously, we get for $x_k \in \var(c) \setminus \var(d)$ that $d \prec c$. Hence $\prec$ is total. 
\end{proof}

\begin{definition}
  For $k\in \{0, \ldots , n\}$, we define the relation $\prec_k \subseteq I \times I$ inductively on $k$ as
  \begin{itemize}
  \item $\prec_0 = \emptyset$
  \item for all $c,d \in I$, $c \prec_{k+1} d$ if and only if $c \prec_k d$ or there exists $e \in I$ such that $c \preceq_k e \prec d$ and $x_{k+1} \in \var(d) \cap \var(e)$,
  \end{itemize}
  where we denote by $c \preceq_k d$ if $c = d$ or $c \prec_k d$.
\end{definition}

Observe that the definition of $\prec_k$ is compatible with the informal discussion of $c \prec_k d$ at the beginning of this section: If $\ck{d}{k}$ depends on $\ck{c}{k}$. For $k=0$, no constraint depends on another, thus $\prec_0 = \emptyset$. Then, when eliminating $x_{k+1}$, if $x_{k+1} \notin \var(d)$, then the dependencies of $d$ do not change since $d$ remains the same. But if $x_{k+1} \in \var(d)$, then the weight of each constraint $e$ whose variables are included in $\var(d)$ and $x_{k+1} \in \var(e)$ will appear in $\ck{d}{k+1}$. And if $e$ depends on $c$ at step $k$, that is $c \prec_k e$, then $d$ will also depend on $c$ after the elimination of $x_{k+1}$.

We now show some properties of $\prec_k$ that are crucial for the understanding of how the weights of constraints interact with each other.

\begin{lemma}
\label{lem:orderprop}
  \begin{enumerate}
  \item[a)] $(\prec_k) \subseteq (\prec_{k+1})$.
  \item[b)] For all $c,d \in I$, $c \prec_{k+1} d$ implies $c \prec d$ and $\var(c) \setminus X_k \subseteq \var(d) \setminus X_k$.
  \item[c)] $(\preceq_k)$ is a partial order.
  \end{enumerate}
\end{lemma}
\begin{proof}
  a)~follows directly from the definition of $\prec_{k+1}$. 
  
  We prove b)~by induction on $k$. For $k=0$, let $c,d \in I$ such that $c \prec_1 d$. Since $\prec_0 = \emptyset$, $c \not \prec_0 d$. Thus, by definition, there exists $e$ such that: $c \preceq_0 e \prec d$ and $x_1 \in \var(c) \cap \var(d)$. Again, since $\prec_0 = \emptyset$, we have $c = e$. Thus $c \prec d$ and since $x_1$ is a nest point, $\var(c) \subseteq \var(d)$, which is the induction hypothesis for $k=0$ since $X_0 = \emptyset$.

Now assume that $k \geq 0$ and that the statement is true for $k$. Let $c,d \in I$ such that $c \prec_{k+1} d$. If $c \prec_k d$, then we get from the induction hypothesis that $c \prec d$ and $\var(c) \setminus X_{k-1} \subseteq \var(d) \setminus X_{k-1}$. This directly yields $\var(c) \setminus X_{k} \subseteq \var(d) \setminus X_{k}$. Now, if $c \not \prec_k d$, then there exists $e$ such that $c \preceq_k e$, $e \prec d$ and $x_{k+1} \in \var(e) \cap \var(d)$. By induction $c \preceq e$ and thus $c \prec d$ since $\prec$ is transitive by Lemma \ref{lem:totalorder}. As $x_{k+1}$ is a nest point after eliminating $X_k$, we have $\var(e) \setminus X_k \subseteq \var(d) \setminus X_k$. By induction we get $\var(c) \setminus X_k \subseteq \var(e) \setminus X_k$ and thus $\var(c) \setminus X_k \subseteq \var(d) \setminus X_k$ as desired.

For c), observe that $\preceq_k$ reflexive by definition. Furthermore, $\preceq_k$ is antisymmetric since it is a subrelation of the order $\prec$ by b). It remains to show that $\prec_k$ is transitive. We do this by induction on $k$. The case $k = 0$ is trivial since $(\prec_0) = \emptyset$. Now suppose that $(\prec_k)$ is transitive for $k \geq 0$. Let $c,d,e \in I$ such that $c \prec_{k+1} d$ and $d \prec_{k+1} e$. If $c \prec_k d \prec_k e$, then by induction $c \prec_k e$ and then $c \prec_{k+1} e$ since $(\prec_k) \subseteq (\prec_{k+1})$. 

Now assume that $c \not \prec_k d$. Then by definition, there exists $c'$ such that $c \preceq_k c' \prec d$ and $x_{k+1} \in \var(c') \cap \var(d)$. Since $d \prec_{k+1} e$, we also have $c' \prec e$ and $x_{k+1} \in \var(d) \setminus X_k \subseteq \var(e) \setminus X_k$. Thus $x_{k+1} \in \var(c') \cap \var(e)$ and $c \preceq_k c' \prec e$, that is $c \prec_{k+1} e$.

Finally assume that $c \prec_k d$ and $d \not \prec_{k} e$. Since $d \prec_{k+1} e$, there exists $d'$ such that $d \preceq_k d' \prec e$ and $x_{k+1} \in \var(d') \cap \var(e)$. By induction, $(\prec_k)$ is transitive. Thus $c \prec_k d' \prec e$ and $x_{k+1} \in \var(d') \cap \var(e)$. That is $c \prec_{k+1} e$. 
\end{proof}

Again, from our intuitive understanding of $\prec_k$, the transitivity is obvious: if $\ck{d}{k}$ depends on $\ck{c}{k}$ and $\ck{e}{k}$ depends on $\ck{d}{k}$, then  $\ck{e}{k}$ should depend on $\ck{c}{k}$. An other informal observation is that if $c$ and $d$ have no common dependencies at step $k$, then they should not share a variable in $X_k$ since sharing a nest point automatically induces a dependency:

\begin{lemma}
\label{lem:cspinter}
  For all $c,d \in I$, if $c \prec d$ but $c \not \prec_k d$, then $\var(c) \cap \var(d) \cap X_k = \emptyset$.
\end{lemma}
\begin{proof}
  By way of contradiction. If for $j \leq k$, $x_j \in \var(c) \cap \var(d) \cap X_k$ then $c \preceq_j c \prec d$ and $x_j \in \var(c) \cap \var(d)$. That is $c \prec_j d$ and by Lemma~\ref{lem:orderprop} we get $c \prec_k d$.
\end{proof}

We need one final property: if $d \prec e$ both depend on $c$ at step $k$, then these dependencies were induced by the elimination of at most two nest points. During the elimination of the second nest point, $e$ will get both the dependencies of $c$ but also the dependencies of $d$. Thus $e$ should depend on $d$. This is formalized by the following lemma:
\begin{lemma}
  \label{lem:cspsup}
  Let $c,d,e \in I$. If $c \prec_k d$, $c \prec_k e$ and $d \prec e$ then $d \prec_k e$.
\end{lemma}
\begin{proof}
  The proof is by induction on $k$. The case $k = 0$ is trivial since the precondition cannot hold. Assume the result holds for $k \geq 0$ and let $c,d,e \in I$ be constraints such that $c \prec_{k+1} d$, $c \prec_{k+1} e$ and $d \prec e$. If both $c \prec_k d$ and $c \prec_k e$, then the induction gives $d \prec_k e$ thus $d \prec_{k+1} e$.

  Otherwise, assume that $c \not \prec_k d$ and $c \not \prec_k e$. Then by definition $x_{k+1} \in \var(d) \cap \var(e)$. Since $d \prec e$, it gives $d \prec_{k+1} e$.

  Now assume $c \not \prec_k d$ but $c \prec_k e$. By definition, there exists $c'$ such that $c \preceq_k c' \prec d$ and $x_{k+1} \in \var(c') \cap \var(d)$. Since $c' \prec d \prec e$, we have $c' \prec e$ and by induction $c \prec_k c'$  and $c \prec_k e$ gives $c' \prec_k e$. Thus $x_{k+1} \in \var(e)$ and $d \prec_{k+1} e$.

  Finally assume that $c \prec_k d$ but $c \not \prec_k e$. By definition, there exists $c'$ such that $c \preceq_k c' \prec e$ and $x_{k+1} \in \var(c') \cap \var(e)$. As in the previous case, by induction, we can deduce that $d \preceq_k c'$ or $c' \prec_k d$. Both cases lead to $d \prec_{k+1} e$.
\end{proof}

We now define for every $k$ and every constraint $c$ a subinstance $I_k(c)$ of $I$ that intuitively contains the relations of $I$ that have an influence on the weights of $c$ after the first $k$ variables have been eliminated. 

\begin{definition}
 For every $k\in \{0, \ldots, n\}$ and $c\in I$ we define $I_k(c) := \{d \in I \mid d \preceq_k c\}$.
\end{definition}

We will now prove a lemma that helps us understand how $I_{k}(c)$ is evolving during the algorithm. Again, the behaviour is intuitively very natural: If $x_{k+1} \notin \var(c)$, then $c$ will have no new dependencies, thus $I_{k+1}(c) = I_k(c)$. If $x_{k+1} \in \var(c)$ however, $c$ will take all the dependencies of the constraints $d$ such that $x_{k+1} \in \var(d)$ and $d \prec c$.

\begin{lemma}
\label{lem:cspik}
  For $k \leq 0$, if  $x_{k+1} \notin \var(c)$ then $I_{k+1}(c) = I_k(c)$. Otherwise, let $I(x_{k+1}) := \{c_1, \ldots, c_m\}$ with $c_1 \prec \ldots \prec c_m$. Then we have
$$ I_{k+1}(c_1) = I_{k}(c_1)$$
and for $i < m$ 
$$ I_{k+1}(c_{i+1}) = I_k(c_{i+1}) \cup I_{k+1}(c_i).$$
\end{lemma}
\begin{proof}
First, assume that $x_{k+1} \notin \var(c)$. Since $(\prec_k) \subseteq (\prec_{k+1})$ by Lemma~\ref{lem:orderprop}, it follows that $I_k(c) \subseteq I_{k+1}(c)$. Now, if $d \in I_{k+1}(c)$ and $d \neq c$, then either $d \prec_k c$ or there exists $e$ such that $d \preceq_k e \prec c$ and $x_{k+1} \in \var(c) \cap \var(e)$. Since $x_{k+1} \notin \var(c)$, we necessarily have $d \prec_k c$, that is $d \in I_k(c)$. This implies $I_k(c) =I_{k+1}(c)$.

For the second equality, $I_{k}(c_1) \subseteq I_{k+1}(c_1)$ still follows from Lemma~\ref{lem:orderprop}. For the other direction, consider $d \in I_{k+1}(c_1)$, that is $d \preceq_{k+1} c_1$. By way of contradiction, assume that $d \not \preceq_{k} c$. By definition of $\prec_{k+1}$, there exists $e \in I$ such that $d \preceq_k e \prec c_1$ and $x_{k+1} \in \var(e)$. However, by definition, $c_1$ is the minimal constraint with respect to $\preceq$ whose variables contain $x_{k+1}$. Thus such an $e\in I$ cannot exist. Consequently, $d \in I_k(c_1)$ and it follows $I_{k+1}(c_1) = I_k(c_1)$.

Now fix $i<m$. By definition of $c_{i+1}$, we have $x_{k+1}\in \var(c_{i+1})$. We first prove that $I_k(c_{i+1}) \cup I_{k+1}(c_i) \subseteq I_{k+1}(c_{i+1})$. By Lemma~\ref{lem:orderprop} again, $I_k(c_{i+1}) \subseteq I_{k+1}(c_{i+1})$. Now let $d \in I_{k+1}(c_i)$. We have $c_i \preceq_k c_i \prec c_{i+1}$ and $x_{k+1} \in \var(c_i) \cap \var(c_{i+1})$ and thus, by definition of $\prec_{k+1}$ this implies $c_i \prec_{k+1} c_{i+1}$. This yields $d \preceq_{k+1} c_i \prec_{k+1} c_{i+1}$ and thus $d \in I_{k+1}(c_{i+1})$. 

Finally, we prove that $I_{k+1}(c_{i+1}) \subseteq I_k(c_{i+1}) \cup I_{k+1}(c_i)$. So let $d \in I_{k+1}(c_{i+1})$. If $d \preceq_k c_{i+1}$, then, by definition, we have $d \in I_k(c_{i+1})$. So assume now that $d \not \preceq_k c_{i+1}$. By definition of $\prec_{k+1}$, there exists $e$ such that $d \preceq_k e \prec c_{i+1}$ and $x_{k+1} \in \var(e) \cap \var(c_{i+1})$. Since $x_{k+1} \in \var(e)$ and $e \prec c_{i+1}$, it follows that $e = c_j$ for a $j < i+1$. If $j = i$, then $d \preceq_k e = c_i$ and thus $d \preceq_{k+1} c_i$ which implies $d \in I_{k+1}(c_i)$. Otherwise, $j < i$ and we have $d \preceq_k c_j \prec c_i$ and $x_{k+1} \in \var(c_j) \cap \var(c_i)$, which gives $d \prec_{k+1} c_i$. Thus $d \in I_{k+1}(c_i)$ as well.
\end{proof}

\subsection{Proof of the runtime bound}

In this section, we will prove that for each $c \in I$ and $a$ an assignment of $\var(\ck{c}{k})$, $\ck{c}{k}(a)$ is proportional to
$$ \twofrac{I_k(c)}{I_k(c) \setminus \{c\}}{a}.$$
Since $I_k(c)$ is a subinstance of $I$, the bitsize of the computed rational number is polynomial in the size of the input. Thus, it will follow that the weight of $\ck{c}{k}$ is a rational number of polynomial bitsize and thus all arithmetic operations of the algorithm can be done in polynomial time. 

Remember that by convention $w(\emptyset, a)=1$ and that for $x \in \var(I)$, $I(x) = \{c \in I \mid x \in \var(c)\}$.

% We are now ready to explain how the cancellations are happening during the elimination process in the algorithm of Section~\ref{sec:cspnest}.
\begin{lemma}
  \label{lem:cspbigprod}
  Let $k \geq 0$ and $I(x_{k+1}) = \{c_1, \ldots, c_m\}$ with $c_1 \prec \ldots \prec c_m$. For all $j \leq m$ and $a : \var(c_j) \setminus X_k \rightarrow D$ we have
$$ \prod_{i=1}^j \twofrac{I_k(c_i)}{I_k(c_i) \setminus \{c_i\}}{a} = \twofrac{I_{k+1}(c_j)}{I_{k+1}(c_j) \setminus \{c_1, \ldots, c_j\}}{a}.$$
\end{lemma}
\begin{proof}
  The proof is by induction on $j$. For $j = 1$, it is a consequence of Lemma~\ref{lem:cspik} since $I_{k+1}(c_1) = I_k(c_1)$. Assume the result holds for $j \geq 1$. Fix $a : \var(c_{j+1}) \setminus X_k \rightarrow D$. Observe first that by Lemma~\ref{lem:orderprop} we have $\var(c_i)\setminus X_k \subseteq \var(c_{j+1}) \setminus X_k$ for $i\le j$ (this could alternatively be seen from the fact that $x_{k+1}$ is a nest point after removing $x_1, \ldots, x_k$). Thus we can use induction for $a$ and get
$$ \prod_{i=1}^{j+1} \twofrac{I_k(c_i)}{I_k(c_i) \setminus \{c_i\}}{a} = \twofrac{I_{k+1}(c_j)}{I_{k+1}(c_j) \setminus \{c_1, \ldots, c_j\}}{a} \twofrac{I_k(c_{j+1})}{I_k(c_{j+1}) \setminus \{c_{j+1}\}}{a}.$$

We will apply Corollary~\ref{cor:cspcancel} with $J_1 := I_{k+1}(c_j)$, $J_2 := I_k(c_{j+1})$, $J_3 := I_k(c_{j+1}) \setminus \{c_{j+1}\}$ and $J_4 := I_{k+1}(c_j) \setminus \{c_1, \ldots, c_j\}$ and $W := \var(c_{j+1}) \setminus X_k$.

Observe that by Lemma~\ref{lem:cspik} and by the fact that $J_3 \subseteq J_2$, $(J_1 \setminus J_3) \cup J_2 = J_1 \cup J_2 = I_{k+1}(c_{j+1})$. Moreover,  $(J_3 \setminus J_1) \cup J_4 = (J_3 \setminus J_1) \cup (J_1 \setminus \{c_1, \ldots, c_j\}) = (J_1 \cup J_3) \setminus \{c_1, \ldots, c_j\} = I_{k+1}(c_{j+1}) \setminus \{c_1, \ldots, c_{j+1}\}$ since $c_{j+1} \notin J_1$. Hence, if the conditions of Corollary~\ref{cor:cspcancel} are met, the lemma will follow.

We now verify each conditions of Corollary~\ref{cor:cspcancel}:

\begin{enumerate}[label=(\roman{enumi})]
\item if $c \in J_1 \cap J_2$, then $c \preceq_{k+1} c_{j+1}$ (it is in $J_2$) and $c \neq c_{j+1}$ since $c \preceq c_j \prec c_{j+1}$. Thus $c \in J_3$. Moreover $J_4 \subseteq J_1$, thus $J_3 \cap J_4 \subseteq J_1$.
\item since $J_3 \subseteq J_2$, this condition is implied by condition (iv).

%  let $c \in J_1 \setminus J_3$ and $d \in J_1 \cap J_3$. We want to show that $\var(c) \cap \var(d) \subseteq \var(c_{j+1}) \setminus X_k = W$. Since both $c,d \in J_1$, we have $c \prec_{k+1} c_j$ and $d \prec_{k+1} c_j$. By Lemma~\ref{lem:orderprop}, $\var(c) \setminus X_k \subseteq \var(c_j) \setminus X_k \subseteq W$ and $\var(d) \setminus X_k \subseteq W$. 

% If $c \prec_k d$, as $d \prec_k c_{j+1}$, $c \prec_k c_{j+1}$ by transitivity, that is $c \in J_3$ which is absurd. If $d \prec_k c$, as $d \prec_k c_{j+1}$, we would have $c \prec_k c_{j+1}$ by Lemma~\ref{lem:cspsup} (and by remarking that $c \prec c_j \prec c_{j+1}$), which is absurd again. Thus $c$ and $d$ are not comparable for $\prec_k$. By Lemma~\ref{lem:cspinter}, we have $\var(c) \cap \var(d) \cap X_k = \emptyset$. Thus $\var(c) \cap \var(d) \subseteq W$. Consequently, $\var(J_1 \setminus J_3) \cap \var(J_1 \cap J_3) \subseteq W$. 

\item Let $c \in J_3 \setminus J_1$ and $d \in J_1$. Since both $c \in J_3$ and $d \in J_1$, we have $c \prec_k c_{j+1}$ and $d \prec_{k+1} c_{j} \prec_{k+1} c_{j+1}$. By Lemma~\ref{lem:orderprop}, $\var(c) \setminus X_k \subseteq W$ and $\var(d) \setminus X_k \subseteq W$.

We claim that $c$ and $d$ are incomparable with respect to $\prec_k$.

First, if $c \prec_k d$, then $c \prec_{k+1} d \prec_{k+1} c_j$ that is $c \in J_1$ which is a contradiction. Consequently, $c\nprec_k d$.

Now, if $d \prec_k c$, then $d \prec_{k+1} c$ and $d \prec_{k+1} c_j$. Thus, since $c \not \preceq_{k+1} c_j$, we have $c_j \prec_{k+1} c \prec_{k+1} c_{j+1}$ thus $c_j \prec c \prec c_{j+1}$. We have that $x_{k+1}\in \var(c_j)$, and by Lemma~\ref{lem:orderprop} b) we get $x_{k+1} \in \var(c)$. But this contradicts the definition of $c_j$ as the maximal constraint with respect to $\preceq$ that is less than $c_{j+1}$ and holds $x_{k+1}$. Hence this is a contradiction and we get $d\nprec c$.

Thus, $c$ and $d$ are indeed incomparable with respect to $\prec_k$. Since $\prec$ is a total order we have either $d\prec c$ or $c\prec d$ and thus by Lemma~\ref{lem:cspinter} we have $\var(c) \cap \var(d) \cap X_k = \emptyset$. Since by Lemma~\ref{lem:orderprop} b) we have that $\var(c)\setminus X_k \subseteq \var(c_{j+1})$ and $\var(d)\setminus X_k \subseteq \var(c_{j+1})$, it follows that $\var(c) \cap \var(d) \subseteq W$. Since this is true for all combinations of $c$ and $d$, it follows that $\var(J_3 \setminus J_1) \cap \var(J_1 \cap J_3) \subseteq W$ as desired.

\item let $c \in J_1 \setminus J_3$ and $d \in J_2$. We have $c \prec_{k+1} c_j$ and $d \preceq_{k} c_{j+1}$. By Lemma~\ref{lem:orderprop}, $\var(c) \setminus X_k \subseteq \var(c_j) \setminus X_k \subseteq W$ and $\var(d) \setminus X_k \subseteq W$. 

We again show that $c$ and $d$ are incomparable with respect to $\prec_k$.

If $c \prec_k d$, we get with $d \preceq_k c_{j+1}$ and transitivity $c \prec_k c_{j+1}$. Thus $c \in J_3$ which is a contradiction. Consequently, $c\nprec_k d$.

Now assume that $d \prec_k c$. We have $c \prec c_j \prec c_{j+1}$ and $d\preceq_k c_{j+1}$ and thus with Lemma~\ref{lem:cspsup} we get $c \prec_k c_{j+1}$. But then $c\in J_3$ which is a contradiction again. 

Thus $c$ and $d$ are indeed incomparable with respect to $\prec_k$. Now the claim follows as in (iii).

\item since $J_4 \subseteq J_1$, this is implied by our proof of condition (iii) (we have not assumed $d \in J_3$ there).
\end{enumerate}
\end{proof}

We can now state the main theorem of this section. Remember that $c^{(k)}$ is the weighted constraint we get from $c$ after $k$ steps of our elimination procedure.

\begin{theorem}
\label{thm:explicit}
  For all $c \in I$ and $k \geq 0$, there exists $\alpha_k(c) \in \N \setminus \{0\}$ such that for all $a : \var(c) \setminus X_k \rightarrow D$, either
$$ \ck{c}{k}(a) = 0$$
or
$$ \ck{c}{k}(a) = \frac{1}{\alpha_k(c)} \cdot \twofrac{I_k(c)}{ I_k(c) \setminus \{c\}}{a}$$
and $\alpha_k(c) \leq |D|^k$.
\end{theorem}
\begin{proof}
  The proof is by induction on $k$. Note that $\prec_0 =\emptyset$ by definition and by convention $w(\emptyset,a) = 1$. So taking $\alpha_0(c) = 1$, proves the result for $k=0$.

  Now assume that the result holds for $k \geq 0$. To lighten the notations, we will denote $x_{k+1}$ by $x$. 
  
  If $x \notin \var(c)$, then $\ck{c}{k} = \ck{c}{k+1}$. By Lemma~\ref{lem:cspik}, we also know that $I_{k+1}(c) = I_k(c)$. Thus, if by choosing $\alpha_{k+1}(c) = \alpha_k(c)$, the result follows.

So consider now $I(x)$, i.e.~the constraints that contain $x$ as a variable. Let $I(x) = \{c_1, \ldots, c_m\}$ with $c_1 \prec \ldots \prec c_m$. We will prove the result for all of the $c_i$ by induction on $i$. For $i=1$, we have by definition, for all $a : \var(c_1) \setminus X_{k+1} \rightarrow D$, either $\ck{c_1}{k+1}(a) = 0$ and there is nothing to prove, or
$$ \ck{c_1}{k+1}(a) = \frac{\sum_{d \in D} \ck{c_1}{k}(a \oplus_x d)}{|D|}.$$
By induction on $k$, we get
$$\ck{c_1}{k+1}(a) = \frac{1}{|D|\alpha_k(c_1)}\sum_{d \in D'} \twofrac{I_k(c_1)}{I_k(c_1) \setminus \{c_1\}}{a \oplus_x d}$$
where $D' = \{d \in D \mid c_1(a \oplus_x d) \neq 0\}$. As there is no constraint in $I_k(c_1) \setminus \{c_1\}$ having the variable $x$, the denominator in the sum does not depends on $d$. Moreover, $I_{k+1}(c_1) = I_k(c_1)$ by Lemma~\ref{lem:cspik}. If $d \notin D'$ then $c_1(a \oplus_x d) = 0$ and hence $w(I_k(c_1),a \oplus_x d) = 0$. Thus, if we set $\alpha_{k+1}(c_1) = |D|\alpha_k(c_1)$, we have
$$
\begin{aligned}
  \ck{c_1}{k+1}(a) & = \frac{\alpha_{k+1}(c_1)^{-1}}{w(I_{k+1}(c_1) \setminus \{c_1\},a)} \sum_{d \in D} w(I_{k+1}(c_1),a \oplus_x d) \\
  & = \frac{1}{\alpha_{k+1}(c_1)} \cdot \twofrac{I_{k+1}(c_1)}{I_{k+1}(c_1) \setminus \{c_1\}}{a}.
\end{aligned}
$$

For $i > 1$, for all $a: \var(c_{i+1}) \setminus X_{k+1} \rightarrow D$, either $\ck{c_{i+1}}{k+1}(a) = 0$ and there is nothing to prove, or by definition
$$ \ck{c_{i+1}}{k+1}(a) = \frac{\sum_{d \in D} \prod_{j \leq i+1} \ck{c_j}{k}((a \oplus_x d)|_{\var(\ck{c_j}{k})})}{\sum_{d \in D} \prod_{j \leq i} \ck{c_j}{k}((a \oplus_x d)|_{\var(\ck{c_j}{k})})}.$$
Applying the induction hypothesis and Lemma~\ref{lem:cspbigprod} on both the numerator and the denominator, by also remarking that $I_k(c_i) \setminus \{c_1, \ldots, c_i\}$ does not contain any constraint with the variable $x$
$$ \ck{c_{i+1}}{k+1}(a) = \frac{1}{\alpha_{k}(c_{i+1})} \cdot \twofrac{I_{k+1}(c_{i+1})}{I_{k+1}(c_i)}{a} \twofrac{I_{k+1}(c_i) \setminus \{c_1, \ldots, c_i\}}{I_{k+1}(c_{i+1}) \setminus \{c_1, \ldots, c_{i+1}\}}{a}. $$

We now apply Corollary~\ref{cor:cspcancel} with $W := \var(c_{i+1}) \setminus X_{k+1}$, $J_1 := I_{k+1}(c_i) \setminus \{c_1, \ldots, c_i\}$, $J_2 := I_{k+1}(c_{i+1})$, $J_3 := I_{k+1}(c_{i+1}) \setminus \{c_1, \ldots, c_{i+1}\}$ and $J_4 := I_{k+1}(c_i)$.
Note that this will yields the desired result: We have $(J_1 \setminus J_3) \cup J_2 = J_2 = I_{k+1}(c_{i+1})$ since $J_1 \subseteq J_3$ and $(J_3 \setminus J_1) \cup J_4 = I_{k+1}(c_{i+1}) \setminus \{c_{i+1}\}$, from combining Lemma~\ref{lem:cspik} and the fact that $\{c_1, \ldots, c_i\} \subseteq J_4$ and $c_{i+1} \notin J_4$.

We now check the conditions of Corollary~\ref{cor:cspcancel}.
\begin{enumerate}[label=(\roman{enumi})]
\item Since $J_1 \subseteq J_3$, we have $J_1 \cap J_2 \subseteq J_3$. Moreover, $J_3 \cap J_4 \subseteq J_1$ since $J_1 = J_4 \setminus \{c_1, \ldots, c_i\}$ and $J_3$ does not contain any of the $c_1, \ldots, c_i$.
\item This condition holds since $J_1 \setminus J_3 = \emptyset$.
\item This condition is a consequence of condition (v) since $J_1 \cap J_3 \subseteq J_4$.
\item This condition holds since $J_1 \setminus J_3 = \emptyset$.
\item Let $c \in J_3 \setminus J_1$ and $d \in J_4$. We have that $c \prec_{k+1} c_{i+1}$. Moreover, $c_i \preceq_k c_i \prec c_{i+1}$ and $x\in \var(c_i)\cap \var(c_{i+1})$ and consequently, by definition of $\prec_{k+1}$, we have $c_i \prec_{k+1} c_{i+1}$. By definition of $J_4$ we have $d \prec_{k+1} c_{i}$ thus by transitivity of $\prec_{k+1}$ we get $d\prec_{k+1} c_{i+1}$. Using Lemma~\ref{lem:orderprop} b), it follows that $\var(c) \setminus X_{k+1} \subseteq W$ and $\var(d) \setminus X_{k+1} \subseteq W$.

Note that $c \not \preceq_{k+1} c_i$, because $c \notin J_1$ and $c \notin \{c_1, \ldots, c_{i+1}\}$.

We now show that $c$ and $d$ are incomparable with respect to $\prec_{k+1}$.

By way of contradiction, assume first that $c \prec_{k+1} d$. Then as $d \prec_{k+1} c_i$, we get $c \prec_{k+1} c_i$ which is a contradiction. 

Now assume that $d \prec_{k+1} c$. With $d \prec_{k+1} c_i$ and the fact that $\prec$ is a total order we get from Lemma~\ref{lem:cspsup} that $c \prec_{k+1} c_i$ or $c_i \prec_{k+1} c$. But we know that  $c \not \prec_{k+1} c_i$, so it follows that $c_i \prec_{k+1} c \prec_{k+1} c_{i+1}$ and thus $c_i \prec c \prec c_{i+1}$. By definition of $c_i$, we have $x\in \var(c_i)$ and by Lemma~\ref{lem:orderprop} it follows that $x\in \var(c)$. But this contradicts the choice of $c_i$ as the maximal element in $I(x)$ with respect to $\prec$ that is less than $c_{i+1}$.

Consequently, $c$ and $d$ are in fact incomparable with respect to $\prec_{k+1}$. Now $(v)$ follows as in as in $(iii)$ in the proof of Lemma~\ref{lem:cspbigprod}.
\end{enumerate}

Having checked all conditions, we may apply Corollary~\ref{cor:cspcancel} which concludes the proof.
\end{proof}

Combining the results of Section~\ref{sec:cspnest} and Section~\ref{sec:analysis}, we now state the main tractability result of this paper.

\begin{theorem}
\label{thm:betaalgo}
There exists an algorithm that, given a $\beta$-acyclic instance $I$ of $\sCSPd$ on domain $D$, computes $w(I)$ in polynomial time. %\florent{ask Johann for accurate bounds}
\end{theorem}
\begin{proof}
In a first step, one computes a $\beta$-elimination order for $\calH(I)$, which can be done naively in polynomial time, iteratively searching by brute force for a nest point. When it is found, we remove the nest point and iterate.

Then we can iterate the elimination procedure of Theorem~\ref{thm:cspalgo}, respecting the order $\prec$ of Section~\ref{sec:analysis} induced by the elimination order. We make $O(s(I)^2\|I\|)$ arithmetic operations to perform all the elimination steps. The other operations needed are the computation of the new supports of the constraints at each step, which can be done in polynomial time. 

Finally, Section~\ref{sec:analysis} provides a good upper bound on the size of the rationals on which we need to perform arithmetic operations. They are always of polynomial bitsize (of size $O(|\var(I)| \log |D|)$), thus each operation can be perform in polynomial time.
\end{proof}

Combining Theorem~\ref{thm:betaalgo} and Corollary~\ref{cor:sattocsp} we get the main tractability result for $\sSAT$.

\begin{corollary}
  $\sSAT$ on $\beta$-acyclic CNF-formulas can be solved in polynomial time.
\end{corollary}

\section{Relation to the STV-framework}

In this section we compare our algorithmic result for $\sSAT$ on $\beta$-acyclic hypergraphs to the framework proposed by S{\ae}ther, Telle and Vatshelle in \cite{SaetherTV14} which we call short the STV-framework. We first show that the STV-framework gives a uniform explanation of all tractability results for $\sSAT$ in the literature, extending the results of \cite{SaetherTV14}. We see this as strong evidence that the STV-framework is indeed a good formalization of the intuitive notion of ``dynamic programming for $\sSAT$''.

Next we show that the STV-framework cannot give any subexponential time algorithms for $\beta$-acyclic $\sSAT$. To this end, we prove an exponential lower bound on the PS-width of $\beta$-acyclic CNF-formulas.

\subsection{Explaining old results by PS-width}

In this section we show that the STV-framework is indeed strong enough to explain all known results on structural $\sSAT$. Figure~\ref{FIG:hierarchy} shows the hierarchy for inclusion formed by the acyclicity notions and classes defined by bounding the width measures from the literature. Most proofs of inclusion can be found in \cite{Fagin83,Duris12,GottlobP01,PaulusmaSS13,CapelliDM14} and the references therein. The relation between disjoint branches and MIM-width and that between $\beta$-acyclicity and MIM-width are shown in this paper.

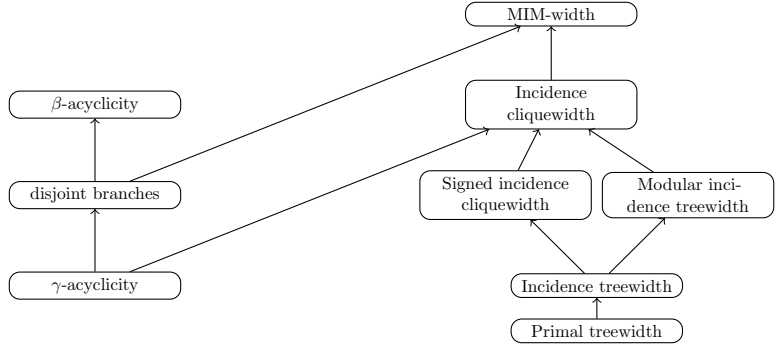
\begin{figure}[t]
\center
\begin{tikzpicture}[scale=0.6, every node/.style={transform shape}, auto]
  \definecolor{sp}{RGB}{100,180,255}
  \definecolor{p}{RGB}{255,255,180}

  \tikzstyle{FPT} = [rectangle,draw,rounded corners,text centered,text width=3.5cm];
  \tikzstyle{XP} = [rectangle,draw,rounded corners,text centered,text width=3.5cm];
  \tikzstyle{sP} = [rectangle,draw,rounded corners,text centered,text width=3.5cm];
  \tikzstyle{P} = [rectangle,draw,rounded corners,text centered,text width=3.5cm];
  \tikzstyle{uK} = [rectangle,draw,rounded corners,text centered,text width=3.5cm];

  \node[P]  (ga)  at (0,0)  {$\gamma$-acyclicity};
  \node[P]  (db) at  (0,2)  {disjoint branches};
  \node[uK] (ba) at  (0,4) {$\beta$-acyclicity};
 % \node[sP] (aa) at  (0,6) {$\alpha$-acyclicity};
 
  %\node[sP] (hw) at  (6,8) {Hypertree-width};

  \node[FPT] (sicw) at (9,2) {Signed incidence cliquewidth};
  \node[XP] (icw) at (10,4) {Incidence cliquewidth};
  \node[XP] (mimw) at (10,6) {MIM-width};
 % \node[sP] (bhw) at (10,6) {$\beta$-hypertree-width}; 
  \node[FPT] (mitw) at (13,2) {Modular incidence treewidth};
  \node[FPT] (itw) at (11,0) {Incidence treewidth};
  \node[FPT] (ptw) at (11,-1) {Primal treewidth};

  \draw[->] (ga) -- (db);
  \draw[->] (db) -- (ba);
%  \draw[->] (ba) -- (aa);
%  \draw[->] (aa) -- (hw);
  \draw[->] (ga) -- (icw);
  \draw[->] (ptw) -- (itw);
  \draw[->] (itw) -- (sicw);
  \draw[->] (itw) -- (mitw);
  \draw[->] (mitw) -- (icw);
  \draw[->] (sicw) -- (icw);
  \draw[->] (db) -- (mimw);
  \draw[->] (icw) -- (mimw);
 % \draw[->] (bhw) -- (hw);
 % \draw[->] (ba) -- (bhw);
\end{tikzpicture}
\caption{A hierarchy of inclusion of graph and hypergraph classes. Classes not connected by a directed path are incomparable. Note that we leave out PS-width because it is not a graph width measure.}~\label{FIG:hierarchy}
\end{figure}

Known complexity results for the restrictions of $\sSAT$ can be found in Table~\ref{table:results}; for definitions of the appearing complexity classes see e.g.~\cite{FlumG06}.

\begin{table}[t]
\begin{center}
\begin{tabular}{l|l|l}
class & lower bound & upper bound\\
\hline
primal treewidth &  & $\FPT$~\cite{Samer:2010wc}\\
incidence treewidth &  & $\FPT$~\cite{Samer:2010wc}\\
modular incidence treewidth & $\W{1}$-hard~\cite{PaulusmaSS13} & $\mathsf{XP}$~\cite{PaulusmaSS13}\\
signed incidence cliquewidth &  & $\FPT$~\cite{FMR-08}\\
incidence cliquewidth & $\W{1}$-hard~\cite{OrdyniakPS13} & $\mathsf{XP}$~\cite{SlivovskyS13}\\
MIM-width & & $\mathsf{XP}$~\cite{SaetherTV14}\\
$\gamma$-acyclic & & $\mathsf{FP}$~\cite{GottlobP01,SlivovskyS13}\\
disjoint branches & & $\mathsf{FP}$~\cite{CapelliDM14} \\
$\beta$-acyclic &  & $\mathsf{FP}$ (this paper) \\
\end{tabular}
\end{center}
\caption{Known complexity results for structural restrictions of $\sSAT$.}
\label{table:results}
\end{table}

In \cite{Vatshelle12} it is shown that MIM-width is bounded by cliquewidth, so nearly all tractability results of Table~\ref{table:results} follow from \cite{SaetherTV14}. To show that the missing results can also be explained in the STV-framework, we only have to recover the tractability results for formulas with disjoint branches decompositions and the fixed-parameter result for formulas of bounded signed incidence cliquewidth. We reprove these results in the following sections by giving upper bounds on the MIM-width and the PS-width, respectively.

\subsubsection{Hypergraphs with disjoint branches}

In this section we show how the tractability of $\sSAT$ on hypergraphs with a disjoint branches decomposition proved in \cite{CapelliDM14} can be explained by the STV-framework.

A \emph{join tree} $(T, \lambda)$ of a hypergraph $\calH = (V,E)$ consists of a rooted tree $T$ and a mapping $\lambda: V(T)\rightarrow E$ such that the following connectivity condition is satisfied: Let $t_1, t_2\in V(T)$ and $v\in \lambda(t_1) \cap \lambda(t_2)$, then $v\in \lambda(t)$ for every $t\in V(T)$ that lies on the path in $T$ connecting $t_1$ and $t_2$. A join tree is a \emph{disjoint branches decomposition} if whenever $t_1$ and $t_2$ lie on different branches of $T$, we have $\lambda(t_1)\cap \lambda(t_2) = \emptyset$. Hypergraphs with disjoint branches decompositions are a strict subclass of $\beta$-acyclic hypergraphs~\cite{Duris12}.

\begin{theorem}\cite{CapelliDM14}\label{thm:computeDB}
 There is an algorithm that, given a hypergraph $\calH$, in time polynomial in $\|\calH\|$ compute a disjoint branches decomposition of $\calH$ if one exists and rejects otherwise.
\end{theorem}

 \begin{lemma}\label{lem:DB}
  Given a hypergraph $\calH$ and a disjoint branches decomposition of $\calH$, we can in polynomial time compute a branch decomposition of $I(G)$ of MIM-width at most~$2$.
 \end{lemma}
\begin{proof}
 Let $(\calT, \lambda)$ be a disjoint branches decomposition of $\calH=(V,E)$. We construct a branch decomposition $(T,\delta)$ of $\calH$ as follows: The vertices of $\calT$ form the internal vertices of $T$. For every $v\in V$ we introduce a new leaf $u$ labeled by $\delta(u)=v$ connecting it to the vertex of $\calT$ that corresponds to the edge containing $v$ that is farthest from the root of $\calT$. Observe that this choice is unique because $\calT$ has disjoint branches and thus vertices $v\in V$ only appear along a path from the root to a leaf. Furthermore, we add a new leaf $u$ for each $e\in E$ labeled by $\delta(u)=e$, connecting it to the vertex $x$ of $\calT$ with $\lambda(x)=e$. 
 
 We now make $T$ subcubic: For any internal vertex $x$, we introduce a binary tree $T_x$ having as leaves the leaf children of $x$ and connect it to $x$. After that, for every vertex $x$ having more than two children, we introduce again a binary tree $T_x'$ having the children of $x$ as its leaves and connect it to $x$. The result is a branch decomposition $(T,\delta)$ of the incidence graph of $\calH$.
 
 We claim that $(T,\delta)$ has MIM-width at most $2$. So let $v$ be a cut vertex with cut $(X,\bar{X})$. First assume that $v$ lies in one of the $T_x$.
%  or is equal to a $x$. 
 Let $e=\lambda(x)$ be the single $e\in E$ that appears as label of a leaf of $T_x$. Observe that all $u\in V\cap X$ lie in $e$. Also, all $u\in V\cap X$ that lie in an edge different from $e$ must lie in a common edge $e'\in E$ that corresponds to the parent of $e$ in $\calT$. 
 Since $e'\notin X$ only one vertex in $X\cap V$ can contribute to an independent matching in $I(\calH)[X, \bar{X}]$. Furthermore, $e$ is the only edge in $E\cap X$, and it follows that the MIM-width of the cut $(X,\bar{X})$ is at most $2$.
 
 If $v$ does not lie in any $T_x$---that is $v$ lies in a $T'_y$ or is a vertex $y\in V(\calT)$---then the cut $(X,\bar{X})$ corresponds to cutting subtrees $\calT_1, \ldots, \calT_s$ from a vertex $x$ in $\calT$. Every vertex $u\in X\cap V$ lies in an edge $e\in X\cap E$ which is the label $\lambda(x')$ for some vertex $x'$ in a $\calT_i$. Now if $u$ is also in an edge $e'\in \bar{X}\cap E$, then $u\in \lambda(x) \in \bar{X}\cap E$. Consequently, only one vertex $u\in X\cap V$ can be an end vertex of an induced matching in $I(\calH)[X, \bar{X}]$. Furthermore, no vertex $u$ in $\bar{X}\cap V$ is in an edge $e\in X\cap E$, because we connected $u$ to the vertex $y$ farthest from the root in the construction of $T$ and thus cutting outside $T_x$ we cannot be in a situation where $u\notin X$. Consequently, the MIM-width of the cut $(X,\bar{X})$ is at most $1$.
\end{proof}

\begin{corollary}[\cite{CapelliDM14}]
 $\sSAT$ on hypergraphs with disjoint branches decompositions can be solved in polynomial time.
\end{corollary}
\begin{proof}
 Given a CNF-Formula $F$, compute a disjoint branches decomposition with Theorem~\ref{thm:computeDB}. Then apply the construction of Lemma~\ref{lem:DB} to get a branch decomposition of MIM-width at most~$2$. Now combining Theorem \ref{thm:MIMvsPS} and Theorem \ref{thm:dynamic} yields the results.
\end{proof}

\subsubsection{Signed incidence cliquewidth}
 
In this section we use the STV-framework to reprove a result from~\cite{FMR-08} stating that $\sSAT$ is fixed-parameter tractable parameterized by signed cliquewidth. We first state the relevant definitions from~\cite{FMR-08}.

The \emph{signed incidence graph} $SI(F)$ of a CNF-formula is the incidence graph of $F$ where each edge $xC$ is signed positively or negatively depending on if the variable $x$ appears positively or negatively in the clause $C$. 
The set of CNF-formulas of signed cliquewidth at most $k$ is defined as the set of formulas whose signed incidence graph can be obtained by the following operations over graphs whose vertices are coloured by $\{1, \ldots , k\}$, starting from singleton graphs.
\begin{enumerate}
 \item Disjoint union.
 \item Recolouring: For a vertex-coloured signed bipartite graph $G$, we defined $\rho_{i,j}(G)$ to be the graph that results from recolouring with $j$ all vertices that were previously coloured with $i$.
 \item Positive edge creation: For a vertex-coloured signed bipartite graph $G$, we define $\eta_{i,j}^+(G)$ to be the graph that results from connecting all clause-vertices coloured $i$ to all variable-vertices coloured $j$, with edges signed positively. We do not add edges between variable vertices coloured $i$ and clause-vertices coloured $j$, or any other vertices.
 \item Negative edge creation: Similarly to above, we define $\eta_{i,j}^-(G)$ to be the graph resulting from connecting all clause-vertices coloured with $i$ to all variable-vertices coloured with $j$, with edges signed negatively.
 \end{enumerate}
The \emph{signed cliquewidth} of a CNF-formula is the minimum $k$ such that it has signed cliquewidth at most~$k$.

A \emph{parse tree} for the signed cliquewidth of a formula $F$ is the rooted tree whose leaves hold singleton graphs, whose internal vertices are coloured with the operations of the definitions above (so a vertex corresponding to a disjoint union has two children, and vertices corresponding to other operations have one child), and whose root holds the graph $SI(F)$ (with any vertex colouring).

Given a signed parse tree of a formula $F$, we construct iteratively a branch decomposition. We assume w.l.o.g.~that whenever we make a union, the graphs whose union we take have only disjoint colors in their vertex coloring. This can be easily achieved by at most doubling the number of colors used. Furthermore, we assume that in the end all vertices have the same color.
 
We construct the branch decomposition along the parse tree iteratively. To this end, we assign a tree $T_\tau$ to each sub-parse tree $\tau$. To a singleton $v$ representing a variable of $F$, we assign a singleton vertex labeled with $v$. For $\tau=\eta_{i,j}^+(\tau')$ and $\tau=\eta_{i,j}^-(\tau')$ we set $T_\tau:= T_{\tau'}$. For $\tau=\rho_{i,j}(\tau')$ we again let $T_\tau := T_{\tau'}$. Finally, for $\tau=\tau_1\cup \tau_2$ we introduce a new root and connect it to $T_{\tau_1}$ and $T_{\tau_2}$. Observe that $T_\tau$ is essentially the tree we get from $\tau$ by forgetting internal labels and contracting all paths to edges. Observe that the result $(T,\delta)$ is obviously a branch decomposition.
 
%  For $\tau=\tau_1\cup \tau_2$ we set $F_\tau := F_{\tau_1} \cup F_{\tau_2}$. Finally, when $\tau=\rho_{i,j}(\tau')$ we introduce a new internal vertex which gets as children the roots of the trees in $F_{\tau'}$ containing the vertices of color $i$ and color $j$. It is easy to see that each tree in $F_{\tau'}$ contains only vertices of one color in the graph defined by $\tau'$, so this last operation joins the two trees in $F_{\tau'}$ that contain all of the vertices of color $i$ and color $j$ in $\tau'$, respectively. Since in the end all vertices will have the same color, the construction obviously gives a branch decomposition $(T, \delta)$.
 
 \begin{lemma}
  $(T, \delta)$ has PS-width at most $2^{2k}$.
 \end{lemma}
\begin{proof}
Let $v$ be a cut vertex with the cut $(A,\bar{A})$. Let $X:= A\cap \var(F)$, $\bar{X}:= \bar{A}\cap \var(F)$, $C:=A\cap \cla(F)$ and $\bar{C} := \bar{A}\cap \cla(F)$. Let $\tau$ be the sub-parse tree which is rooted by the union that led to the introduction of $v$.

We first show that $|PS(F_{X,\bar{C}})|\le 2^{2k}$. Observe that when two variables $x, x'\in X$ have the same color in $\tau$, then they must always appear together in every clause in $\bar{C}$ and their sign must be the same. Call $X_i$ the set of variables in $X$ that are colored by $i$. Then for every assignment of $F_{X,\bar{C}}$ the set of satisfied clauses depends only on if there is a variable in $X_i$ that is set to true if $X_i$ appears positively or if there is a variable in $X_i$ set to false if $X_i$ appears negatively. So to get the same precise satisfiability set, we can delete all but two variables from $X_i$ from $F_{X, \bar{C}}$. It follows that $F_{X,\bar{C}}$ has the same precise satisfiability set as a formula with $2k$ variables. But there are only $2^{2k}$ assignments to $2k$ variables, so it follows that $|PS(F_{X,\bar{C}})|\le 2^{2k}$.

We now show that $|PS(F_{\bar{X},C})|\le 2^{2k}$. To this end observe that if two clauses $C,C'$ in $\tau$ have the same color $i$, then they will contain the same variables in $\bar{X}$ and moreover $C|_{\bar{X}} = C'|_{\bar{X}}$. Thus $F_{\bar{X},C}$ only has $k$ different clauses, so trivially $|PS(F_{\bar{X},C})|\le 2^{2k}$.
\end{proof}

\begin{corollary}[\cite{FMR-08}]\label{cor:FMR}
 $\sSAT$ on formulas of signed incidence cliquewidth $k$ can be solved in time $2^{O(k)} |F|^2$ assuming that we are provided a parse tree of width $k$.
\end{corollary}

Note that the runtime bound in~\cite{FMR-08} cannot be easily compared, because the runtime in \cite{FMR-08} depends on the size of the parse tree directly and not on the formula. But both results are fixed-parameter results that singly exponentially depend on $k$, so they are at least very close.

\subsection{Lower bounds on MIM-width and PS-width}

In this section we will prove the promised lower bound on the PS-width of $\beta$-acyclic CNF-formulas.
We start off with a simple Lemma that can be seen as a partial reverse of Lemma~\ref{thm:MIMvsPS}. We remind the reader that a CNF-formula $F$ is called \emph{monotone} if all variables appear only positively in $F$.

\begin{lemma}\label{lem:logwidth}
 For every bipartite graph $G$ there is a monotone CNF-formula $F$ such that $F$ has the incidence graph $G$ and $\psw(F)\ge 2^{\mimw(G)/2}$.
\end{lemma}
\begin{proof}
 We construct $F$ by choosing arbitrarily one color class of $G$ to represent clauses and the other one to represent variables. This choice then uniquely yields a monotone formula where a clause $C$ contains a variable $x$ if and only if $x$ is connected to $C$ by an edge in $G$.
 
 Let $(T,\delta)$ be a branch decomposition of $G$ and $F$. Let $t$ be a vertex of $T$ with cut $(A,\bar{A})$. Set $X:=\var(F)\cap A$, $\bar{X}:= \var(F)\cap \bar{A}$, $\calC := \cla(F)\cap A$ and $\bar{\calC} := \cla(F)\cap \bar A$. Moreover, let $M$ be a maximum independent matching of $G[A,\bar A]$ and let $V_M$ be the end vertices of $M$.
 
 First assume that $|\calC \cap V_M| \ge |\bar{\calC} \cap V_M|$. Let $C_1, \ldots, C_k$ be the clauses in $\calC \cap V_M $ and let $x_1,\ldots, x_k$ be variables in $\bar{X} \cap V_M$. Note that $k\ge |M|/2$. Since $M$ is an independent matching, every clause $C_i$ contains exactly one of the variables $x_j$, and we assume w.l.o.g.\ that $C_i$ contains $x_i$. Let $a$ be an assignment to the $x_i$ and let $a'$ be the extended assignment of $\bar{X}$ that we get by assigning $0$ to all other variables. Then $a'$ satisfies in $F_{\bar{X},C}$ exactly the clauses $C_i$ for which $a(x_i)=1$ since the formula is monotone. Since there are $2^k$ assignments to the $x_i$, we have $|PS(F_{\bar{X},C})| \ge 2^k \ge 2^{|M|/2}$.
 
 For $|\calC \cap V_M| \le |\bar{\calC} \cap V_M|$ it follow symmetrically that $|PS(F_{X,\bar{C}})| \ge 2^{|M|/2}$. 
 
 Consequently, we have in either case that the PS-width of $F$ is at least $2^{|M|/2}$ and the claim follows.
 \end{proof}

To a graph $G=(V,E)$ we define a graph $G'=(V', E')$ as follows: 
\begin{itemize}
 \item for every $v\in V$ there are two vertices $x_v,y_v\in V'$,
 \item for every edge $e=uv\in E$ there are four vertices $p_{e,u},q_{e,u}, p_{e,v}, q_{e,v} \in V'$,
 \item every $u,v\in V$ we add the edge $x_v y_u$ to $E'$, and
 \item for every edge $e=uv\in E$ we add the edges $p_{e,u}q_{e,u}, p_{e,v}q_{e,v}, x_up_{e,u}, y_vq_{e,u}, x_vp_{e,v}, y_uq_{e,v}$. 
\end{itemize}
These are all vertices and edges of $G'$.

\begin{lemma}\label{lem:bipchordal}
 $G'$ is chordal bipartite.
\end{lemma}
\begin{proof}
We have to show that every cycle $C$ in $G'$ of length at least $6$ has a chord. We consider two cases: Assume first that $C$ contains no vertex $p_{e,v}$ and consequently no $q_{e,v}$ either. Then all vertices of $C$ are $x_v$ or $y_v$ and so $C$ is a cycle in the complete bipartite graph induced by the $x_v$ and $y_v$. Clearly, $C$ has a chord then.
 
 Now assume that $C$ contains a vertex $p_{e,v}$ and consequently also $q_{e,v}$. Let $e=uv$. Then $C$ must also contain $x_v$ and $y_u$, so $x_vy_u\in E'$ is a chord.
\end{proof}

\begin{lemma}\label{lem:mimwidthhigh}
 Let $G$ be bipartite. Then $\tw(G)\le 6 \mimw(G')$.
\end{lemma}
\begin{proof}
 Let $(T',\delta')$ be a branch decomposition of $G'$. Let $A,B\subseteq V(G)$ be the two colour classes of $G$. We construct a branch decomposition $(T,\delta)$ of $G$ by deleting the leaves labeled with $p_{e,u},q_{e,u}, p_{e,v}, q_{e,v}$, and those labeled $x_v$ for $v\in A$ or with $y_v$ for $v\in B$. Then we delete all internal vertices of of $T'$ that have become leaves by these deletions until we get a branch decomposition $T$ with the leaves $x_v$ for $v\in B$ and $y_v$ for $v\in A$. For the leaves of $T$ we define $\delta(t):=v$ where $v\in V$ is such that $\delta'(t)= x_v$ or $\delta'(t)=y_v$. The result $(T,\delta)$ is a branch decomposition of $G$.
 
 Let $t$ be a vertex of $T$ with the corresponding cut $(X,\bar X)$. Let $M\subseteq E$ be a matching in $G[X,\bar{X}]$. Let $(X',\bar{X}')$ be the cut of $t$ in $(T',\delta')$. Let $e=uv\in M$, then $x_u$ and $y_v$ are on different sides of the cut $X'$ and they are connected by the path $x_u p_{e,u}q_{e,u} y_v$. Consequently, there is at least one edge along this path in $G'[X',\bar{X}']$. Choose one such edge arbitrarily. 
 
 Let $M'$ be the set of edges we have chosen for the different edges in $M$. Let $M_x'$ be the set of edges in $M'$ that do not have an end vertex $y_v$ and let $M_y'$ be the set of edges in $M'$ that do not have an end vertex $x_v$.  Let $M''$ be the bigger of these two sets. Since $e'\in M'$ can only have an end vertex $x_v$ or $y_u$ but not both, we have $|M_x'| + |M_y'| \ge |M'|$ and thus $|M''| \ge |M'|/2$.
 
 We claim that $M''$ is an independent matching in $G'$. Clearly, $M'$ is a matching because $M$ is one. Consequently, $M''\subseteq M'$ is also a matching. We now show that $M''$ is also independent. By way of contradiction, assume this were not true. Then there must be two adjacent vertices $u,v\in V'$ that are end vertices of edges in $M''$ but not in the same edge in $M''$. If $u=p_{e',w}$ for some $e'\in E$ and $w\in V$, then $v$ must be $x_w$. But then by construction of $M'$, the vertex $w$ must be incident to two edges in $M$ which contradicts $M$ being a matching. Similarly, we can rule out that $v$ is $q_{e,w}$. Thus, $u$ must be $x_w$ or $y_w$ and $v$ must be $x_{w'}$ or $y_{w'}$. Since $x_w$ and $x_{w'}$ are in the same colour class of $G'$, they are not adjacent. Similarly $y_w$ and $y_{w'}$ are not adjacent. Consequently, we may assume that $u=x_w$ and $v= y_{w'}$. But then they cannot both be an endpoint of an edge in $M''$ by construction of $M''$. Thus $M''$ is independent.
 
 By Lemma~\ref{lem:MMvsTW} we know that there is a $t\in T$ with cut $(X,\bar{X})$ such that  we can find a matching $M$ of size at least $\frac{\tw(G)}{3}$ in $G[X,\bar{X}]$. By the construction above the corresponding cut $(X',\bar X')$ yields an independent matching of size $\frac{\tw(G)}{6}$ in $G'[X', \bar{X}']$. This completes the proof.
\end{proof}

Using the connection between vertex expansion and treewidth (see \cite{GroheM09}) the following lemma is easy to show.

\begin{lemma}\label{lem:expander}
 There is a family $\mathcal{G}$ of graphs and constants $c>0$ and $d\in \mathbb{N}$ such that for every $G\in \mathcal G$ the graph $G$ has maximum degree $d$ and we have $\tw(G)\ge c |E(G)|$.
\end{lemma}

\begin{corollary}\label{cor:chordalbipmim}
 There is a family $\mathcal G'$ of chordal bipartite graphs and a constant $c$ such that for every graph $G\in \mathcal G$ we have $\mimw(G)\ge c |V(G)|$.
\end{corollary}
\begin{proof}
 Let $\mathcal G$ be the class of Lemma~\ref{lem:expander}. We first transform every graph $G \in \mathcal G$ into a bipartite one $G_1$ by subdividing every edge, i.e.~by introducing for each edge $e = uv$ a new vertex $w_e$ and by replacing $e$ by $u w_e$ and $w_e v$. It is well-known that subdividing edges does not decrease the treewidth of a graph (see e.g.~\cite{Diestel05}), and thus $\tw(G)\le \tw(G_1)$. Moreover, $|E(G_1)| = 2|E(G)|$, and thus $\tw(G_1) \geq \frac{1}{2}c|E(G_1)|$. Now let $\mathcal G' = \{G_1' \mid G \in \mathcal G\}$. Then the graphs in $\mathcal G'$ are chordal bipartite by Lemma~\ref{lem:bipchordal} and the bound on the MIM-width follows by combining Lemma~\ref{lem:expander} and Lemma~\ref{lem:mimwidthhigh}.
\end{proof}

We can now easily prove the main result of this section.

\begin{corollary}
 There is a family of monotone $\beta$-acyclic CNF-formulas of PS-width~$2^{\Omega(n)}$ where $n$ is the number of variables in the formulas.
\end{corollary}
\begin{proof}
 Let $\mathcal F$ be the class of monotone CNF-formulas having the class $\mathcal G'$ of Corollary \ref{cor:chordalbipmim} as its incidence graphs. By Theorem~\ref{thm:chorbipbeta} the formulas in $\mathcal F$ are $\beta$-acyclic. Combining the bound on the MIM-width of $G'$ with Lemma~\ref{lem:logwidth} then directly yields the result.
\end{proof}

It follows that the STV-framework cannot prove subexponential runtime bounds for $\sSAT$ on $\beta$-acyclic formulas.

 \section{Conclusion}

We have shown that $\beta$-acyclic $\sSAT$ can be solved in polynomial time, a question left open in \cite{CapelliDM14}. Our algorithm does not follow the dynamic programming approach that was used in all other structural tractability results that were known before, and as we have seen this is no coincidence. Instead, $\beta$-acyclic $\sSAT$ lies outside the STV-framework of \cite{SaetherTV14} that explains all old results in a uniform way.

We close this paper with several open problems that we feel should be explored in the future. First, our algorithm for $\sSAT$ is specifically designed for the case of $\beta$-acyclic formulas, but we feel that the techniques developed, in particular those of Section~\ref{sec:analysis}, might possibly be extended to other classes of hypergraphs that one can characterize by elimination orders. In this direction, it would be interesting to see if hypergraphs of bounded $\beta$-hypertree width, a width measure generalizing $\beta$-acyclicity proposed in~\cite{GottlobP01}, can be characterized by elimination orders and if such a characterization can be used to solve $\sSAT$ on the respective instances. Note that this case lies outside of the STV-framework, therefore dynamic programming without new ingredients is unlikely to work. Also, even the complexity of deciding $\SAT$ on instances of bounded $\beta$-hypertree width is an open problem~\cite{OrdyniakPS13}. 

It might also be interesting to generalize our algorithm to solve cases for which we already have polynomial time algorithms. For example, is there any uniform explanation for tractability of bounded cliquewidth $\sSAT$ and $\beta$-acyclic $\sSAT$, similarly to the way in which the framework of \cite{SaetherTV14} explains tractability for all previously known results?

Finally, we feel that, although we have shown that the STV-framework does not explain all tractability results for $\sSAT$, it is still a framework that should be studied in the future. We believe that there are still many classes to be captured by it in the future and thus we see a better understanding of the framework as an important goal for future research. One question is the complexity of computing branch decompositions of (approximately) minimal MIM-width or PS-width. Alternatively, one could try to find more classes of bipartite graphs for which one can efficiently compute branch decompositions of small MIM-width. This would then directly extend the knowledge on structural classes of CNF-formulas for which dynamic programming can efficiently solve $\sSAT$.

\bibliography{betabiblio}

\begin{appendix}
\section{Extension to $\mSAT$}

The algorithm described in this paper can also be turned into an algorithm for $\mCSPd$---the problem of computing, given a set of weighted constraints $I$, the value $m(I) = \max \{ \prod_{c \in I} c(a|_{\var(c)}) \mid a \in D^{\var(I)}\}$. We first show that we can use $\mCSPd$ to solve $\mSAT$, the problem of computing the maximum number of clauses of a CNF-formula $F$ that can be satisfied simultaneously. 

\begin{lemma}
\label{lem:maxtocsp}
  Given a CNF-formula $F$, one can compute in polynomial time a set $I$ of weighted constraints with default values on variables $\var(F)$ and domain $\{0,1\}$ such that
  \begin{itemize}
  \item $\calH(F) = \calH(I)$,
  \item for all $a \in \{0,1\}^{\var(F)}$, $m(I,a) = 2^{s}$ where $s = |\{C \in F \mid a \models C\}|$, and
  \item $s(I) = \|I\| = |F|$
  \end{itemize}
\end{lemma}
\begin{proof}
  For each clause $C$ of $F$, we define a constraint $c$ with default value $2$ whose variables are the variables of $C$ and such that $\supp(c) = \{a\}$ and $c(a) = 1$, where $a$ is the only assignment of $\var(C)$ that is not a satisfying assignment to $C$. It is easy to check that this construction has the above properties.
\end{proof}

\begin{corollary}
  $\mSAT$ is polynomial time reducible to $\mCSPd$. Moreover, $\mSAT$ restricted to $\beta$-acyclic formulas is polynomial time reducible to $\mCSPd$ restricted to $\beta$-acyclic instances.
\end{corollary}
\begin{proof}
  We transform a CNF-formula $F$ into an instance $I$ of $\mCSPd$ using Lemma~\ref{lem:maxtocsp}. We have $\mCSPd(I) = 2^s$ where $s = \mSAT(F)$, so it just remains to take the logarithm in base~$2$.
\end{proof}

We now show how to adapt our algorithm for $\sCSPd$ to $\mCSPd$.
\begin{theorem}
\label{thm:maxcspalgo}
  Let $I$ be a set of weighted constraints on domain $D$ and $x$ a nest point of $\calH(I)$. Let $I(x) = \{c_1, \ldots, c_p\}$ with $\var(c_1) \subseteq \ldots \subseteq \var(c_p)$. Let $I' = \{c' \mid c \in I\}$ where
  \begin{itemize}
  \item if $c \notin I(x)$ then $c' := c$
  \item if $c = c_i$, then $c_i' := (f_i',\d)$ is the weighted constraint on variables $\var(c) \setminus \{x\}$, with default value $\d(c_i)$ and $\supp(c_i') := \{a \in D^{Y \setminus \{x\}} \mid \exists d \in D, (a \oplus_x d) \in \supp(c)\}$. Moreover, for all $a \in \supp(c_i')$, let $P_i(a,d) := \prod_{j=1}^i c_j((a \oplus_x d)|_{\var(c_j)})$ and $P_0(a,d) = 1$. We define:
$$ f_i'(a) := \frac{\max_{d \in D} P_{i}(a,d)}{\max_{d \in D} P_{i-1}(a,d)}$$
if $\max_{d \in D} P_{i-1}(a,d) \neq 0$ and $f_i'(a) := 0$ otherwise.
\end{itemize}
Then $\calH(I') = \calH(I) \setminus x$, $\|I'\| \leq \|I\|$ and $m(I) = m(I')$. Moreover, one can compute $I'$ with a $O(p\|I(x)\|)$ arithmetic operations.
\end{theorem}
\begin{proof}
  The proof is analogous to that of Theorem~\ref{thm:cspalgo}. Remark that the special case where $\max_{d \in D} P_{i-1}(a,d) = 0$ follows similarly to there since $\max_{d \in D} P_{i-1}(a,d) = 0$ implies that for all $d$ we have $P_{i-1}(a,d) = 0$.
\end{proof}

Now remark that $\max$ is commutative, associative, that is $\max(a, \max(b,c)) = \max (\max(a, b), c)$ and that  distributes with multiplication since all numbers are positive, that is $\max (ab,ac) = a \max(b,c)$. Moreover, we have 
$$\max(\frac{a}{b}, \frac{c}{d}) = \frac{\max(ad,cb)}{bd}.$$

Thus, the results of Section~\ref{sec:techlem} can be adapted in a straightforward fashion and the results of \ref{sec:partorder} still hold. We can now adapt Theorem~\ref{thm:explicit} (we use adapted notations for $m(J,a)$ for $J \subseteq I$ and $a$ a partial assignment).
\begin{theorem}
  For all $c \in I$ and $k \geq 0$, for all $a : \var(c) \setminus X_k \rightarrow D$, either
$$ \ck{c}{k}(a) = 0$$
or
$$ \ck{c}{k}(a) = \twofracm{I_k(c)}{ I_k(c) \setminus \{c\}}{a}.$$
\end{theorem}

Now the tractability results for $\mCSPd$ and $\mSAT$ follow directly.

\begin{theorem}
There is an algorithm that, given a $\beta$-acyclic instance $I$  of $\mCSPd$, computes $m(I)$ in polynomial time.
\end{theorem}
\begin{theorem}
There is an algorithm that solves $\mSAT$ on $\beta$-acyclic CNF-formulas in polynomial time.
\end{theorem}

\end{appendix}

\end{document}